\documentclass[twoside,11pt]{article}
\usepackage{jmlr2e}

\usepackage{amsmath,amsfonts,bm,color}
\usepackage{graphics}
\usepackage{graphicx}
\usepackage{color}
\usepackage{subfigure}
\usepackage[pdftex]{hyperref}

\newcommand{\perm}{{\rm perm}}

\ShortHeadings{Approximating the Permanent with FBP}{Chertkov and Yedidia}
\firstpageno{1}

\begin{document}

\title{Approximating the Permanent\\ with Fractional Belief Propagation}

\author{\name Michael Chertkov \email chertkov@lanl.gov
\addr Theory Division \& CNLS, LANL, Los Alamos, NM, 87545 USA
\AND
\name Adam B.\ Yedidia
\email adamy@mit.edu
\addr MIT, 77 Massachusetts Ave. Cambridge, MA 02139 USA\\
and Theory Division \& CNLS, LANL, Los Alamos, NM, 87545 USA
}

\editor{Submitted to the Journal of Machine Learning Research}
\maketitle

\begin{abstract}
We discuss schemes for exact and approximate computations of permanents, and compare them with each other. Specifically, we analyze the Belief Propagation (BP) approach and its Fractional Belief Propagation (FBP) generalization for computing the permanent of a non-negative matrix. Known bounds and conjectures are verified in experiments, and some new theoretical relations, bounds and conjectures are proposed. The Fractional Free Energy (FFE) functional is parameterized by a scalar parameter $\gamma\in[-1;1]$,  where $\gamma=-1$ corresponds to the BP limit and $\gamma=1$ corresponds to the exclusion principle (but ignoring perfect matching constraints) Mean-Field (MF) limit. FFE shows monotonicity and continuity with respect to $\gamma$. For every non-negative matrix, we define its special value $\gamma_*\in[-1;0]$ to be the $\gamma$ for which the minimum of the $\gamma$-parameterized FFE functional is equal to the permanent of the matrix, where the lower and upper bounds of the $\gamma$-interval corresponds to respective bounds for the permanent. Our experimental analysis suggests that the distribution of $\gamma_*$ varies for different ensembles but $\gamma_*$ always lies within the $[-1;-1/2]$ interval. Moreover, for all ensembles considered the behavior of $\gamma_*$ is highly distinctive, offering an emprirical practical guidance for estimating permanents of non-negative matrices via the FFE approach.
\end{abstract}

\begin{keywords}
Permanent, Graphical Models, Belief Propagation, Exact and Approximate Algorithms,
Learning Flow Models
\end{keywords}

\section{Introduction}

This work is motivated by computational challenges associated with learning stochastic flows from two consecutive snapshots/images of $n$ identical particles immersed in a flow \citep{08CKV,10CKKVZ}. The task of learning consists in maximizing the permanent of an $n\times n$ matrix, with elements constructed of probabilities for a particle in the first image to correspond to a particle in the second image,  over the low-dimensional parametrization of the reconstructed flow. The permanents in this enabling application are nothing but a weighted number of perfect matchings relating particles in the two images.

Inspired by this ``learning the flow" application, we continue in this manuscript the thread of \citet{10WC} and focus on computations of positive permanents of non-negative matrices constructed from probabilities. The exact computation of the permanent is difficult, i.e., it is a problem of likely exponential complexity, with the fastest known general algorithm for computing the permanent of a full $n\times n$ matrix based on the formula from \citet{63Rys} requiring ${\cal O}(n 2^n)$ operations. In fact,  the task of computing the permanent of a non-negative matrix was one of the first problems established to be in the \#-P  complexity class,  and the task is also complete in the class \citep{79Val}.

Therefore, recent efforts have mainly focused on developing approximate algorithms. Three independent developments, associated with the mathematics of strict bounds, Monte-Carlo sampling, and Graphical Models, contributed to this field.

On the ``mathematics of permanents" side,  the emphasis was on establishing rigorous lower and upper bounds for permanents. Many significant results in this line of research are related to the conjecture of \citet{26vdW} that the minimum of the permanent over doubly stochastic matrices is $n!/n^n$, and it is only attained when all entries of the matrix are $1/n$. The conjecture remained open for over 50 years before \citet{81Fal} and \citet{81Ego} proved it. Recently, \citet{08Gur} found an alternative, surprisingly short and elegant proof that also allowed for a number of unexpected extensions. (See e.g. the discussion of \citet{09LS}.)

On the ``Monte-Carlo sampling" side, a very significant breakthrough was
achieved with the invention of the Fully Polynomial Randomized Algorithmic schemes (FPRAS) for the permanent problem \citep{04JSV}: the permanent is approximated in polynomial time, provably with high probability and within an arbitrarily small relative error. The complexity of the FPRAS of \citet{04JSV} is $O(n^{11})$ in the general case. Even though the scaling was improved to $O(n^4\log n)$ in the case of very dense matrices \citep{08HL}, the approach is still impractical for the majority of realistic applications.

On the ``Graphical Model" (GM) side, Belief Propagation (BP) heuristics showed surprisingly good performance \citep{08CKV,09HJ,10CKKVZ}. The BP family of algorithms, originally introduced in the context of error-correction codes by \citet{63Gal}, artificial intelligence by \citep{88Pea}, and related to some early theoretical work in statistical physics by \citet{35Bet}, and \citet{36Pei} on tree graphs, can generally be stated for any GM  according to \citet{05YFW}. The exactness of the BP on any tree, i.e., on a graph without loops, suggests that the algorithm can be an efficient heuristic for evaluating the partition function, or for finding a Maximum Likelihood (ML) solution of the graphical model (GM) defined on sparse graphs. However, in the general loopy cases, one would normally not expect BP to work very well, making the heuristic results of \citet{08CKV,09HJ,10CKKVZ} somehow surprising, even though not completely unexpected in view of the existence of polynomially efficient algorithms for the ML version of the problem \citep{55Kuh,92Ber}, which were shown by \citep{08BSS} to be equivalent to an iterative algorithm of the BP type. This raises questions about understanding the performance of BP. To address this challenge \citet{10WC} established a theoretical link between the exact permanent and its BP approximation. The permanent of the original non-negative matrix was expressed as a product of terms, including the BP-estimate and another permanent of an auxiliary matrix,  $\beta.*(1-\beta)$ \footnote{Here and below we will follow Matlab notations for the component-wise operations on matrices,  such as $A.*B$ for the component-wise, Hadamard, product of the matrices $A$ and $B$.},  where $\beta$ is the doubly stochastic matrix of marginal probabilities of links between particles in the two images (edges in the underlying GM) calculated using the BP approach. (See Theorem \ref{Prop:Perm+BP}.) The exact relation of \citet{10WC} followed from the general Loop Calculus technique of \citet{06CCa,06CCb},  but it also allowed a simple direct derivation. Combining this exact relation with aforementioned results from the mathematics of permanents  led to new lower and upper bounds for the original permanent. Moreover this link between the math side and the GM side  gained a new level of prominence with the recent proof by \citet{11Gur} of the fact that the variational formulation of BP in terms of the Bethe Free Energy (BFE) functional, discussed earlier by \citet{08CKV,10CKKVZ,10WC}, and shown to be convex by \citet{11Von}, gives a provable lower bound to the permanent. Remarkably, this proof of Gurvits was based on an inequality suggested earlier by \citet{98Sch} for the object naturally entering the exact, loop calculus based, BP formulas, $\perm(\beta.*(1-\beta))$.

This manuscript contributes two-fold, theoretically and experimentally, to the new synergy developing in the field. Theory-wise, we generalize the BP approach to approximately computing permanents,  suggesting replacing the BFE functional by its fractional generalization in the general spirit of \citet{Wiegerinck_Heskes_2003} differing from the BFE functional of \citet{05YFW} in the entropy term,  and then derive new exact relations between the original permanent and the results of the fractional approach (see Theorem \ref{theorem:m-exact}). The new object, naturally appearing in the theory, is $\perm (\beta.*(1-\beta).^{-\gamma})$,  where $\gamma\in[-1;1]$,  with $\gamma=-1$ corresponding to BP and $\gamma=1$ corresponding to the so-called exclusion principle (Fermi), but ignoring perfect matching constraints, Mean Field (MF) approximation discussed earlier by \citet{08CKV}. Utilizing recent results from the ``mathematics of permanents,"  in particular from \citet{11Gur}, we show,  that considered as an approximation,  the fractional estimate of the permanent is a monotonic continuous function of the parameter $\gamma$ with $\gamma=-1$ and $\gamma=0$ setting, respectively, the lower bound (achievable on trees)  and an upper bound. We also analyze existing and derive new lower and upper bounds. We adopt for our numerical experiments the so-called Zero-suppressed binary Decision Diagrams (ZDDs) approach of \citet{93Minato} (see e.g. \citet{Knuth:2009:ACP:v4f1}), which outperforms Ryser's formula for realistic (sparsified) matrices, for exactly evaluating permanents, develop numerical schemes for efficiently evaluating the fractional generalizations of BP,  test the aforementioned lower and upper bounds for different ensembles of matrices and study the special, matrix dependent, $\gamma_*$, which is defined to be the special $\gamma$ for which the fractional estimate is equal to the permanent of the matrix \footnote{Note that a methodologically similar approach,  of searching for the best/special fractional coefficient, was already discussed in the literature by \citet{11CH} for a Gaussian BP example.}.

The material in the manuscript is organized as follows: the technical introduction, stating the computation of the permanent as a Graphical Model,  is explained in Section \ref{sec:intro} and Appendix \ref{sec:ML}. The BP-based optimization formulations, approximate methods, iterative algorithms and related exact formulas are discussed in Section \ref{sec:BP+MF} and Appendices \ref{sec:BP}, \ref{sec:MF}, \ref{sec:fractional}, \ref{sec:BP_low}. Section \ref{sec:what_to_test} is devoted to permanental inequalities, discussing the special values of $\gamma$ and conjectures. Our numerical experiments are presented and discussed in Section \ref{sec:Experiments} and Appendices \ref{sec:pruning}, \ref{sec:ZDD}, \ref{sec:Ryser_vs_ZDD}. We conclude and discuss the path forward in Section \ref{sec:conclusion}.

\section{Technical Introduction}

\label{sec:intro}

The permanent of a square matrix $p$, $p=(p_{ij}| i,j=1,\ldots,n)$, is defined as
$$
\perm(p) = \sum_{s\in S_{n}} \prod_{i=1}^n p_{i s(i)},
$$
where $S_{n}$ is the set of all permutations of the set, $\{1,\ldots,n\}$. Here and below we will only discuss permanents of non-negative matrix,  with $\forall i,j=1,\ldots,n: p_{ij}\geq 0$, also assuming that $\perm(p)>0$.

An example of a physics problem, where computations of permanents  are important, is given by particle tracking experiments and measurements techniques, of the type discussed in \cite{08CKV,10CKKVZ}. In this case an element of the matrix, $p=(p_{ij}|i,j=1,\ldots,n)$, is interpreted as an unnormalized probability that the particle labeled $i$ in the first image moves to the position labeled $j$ in the second image. In its most general formulation,  the task of learning a low dimensional parametrization of the flow from two consecutive snapshots consists of maximizing the partition function $Z=\perm(p)$ over the ``macroscopic" flow parameters affecting $p$. Computing the permanent for a given set of values of the parameters constitutes an important subtask,  the one we are focusing on in this manuscript.

\subsection{Computation of the Permanent as a Graphical Model Problem}
\label{subsec:GM}

\begin{figure}
\centering
\includegraphics[width=0.3\textwidth]{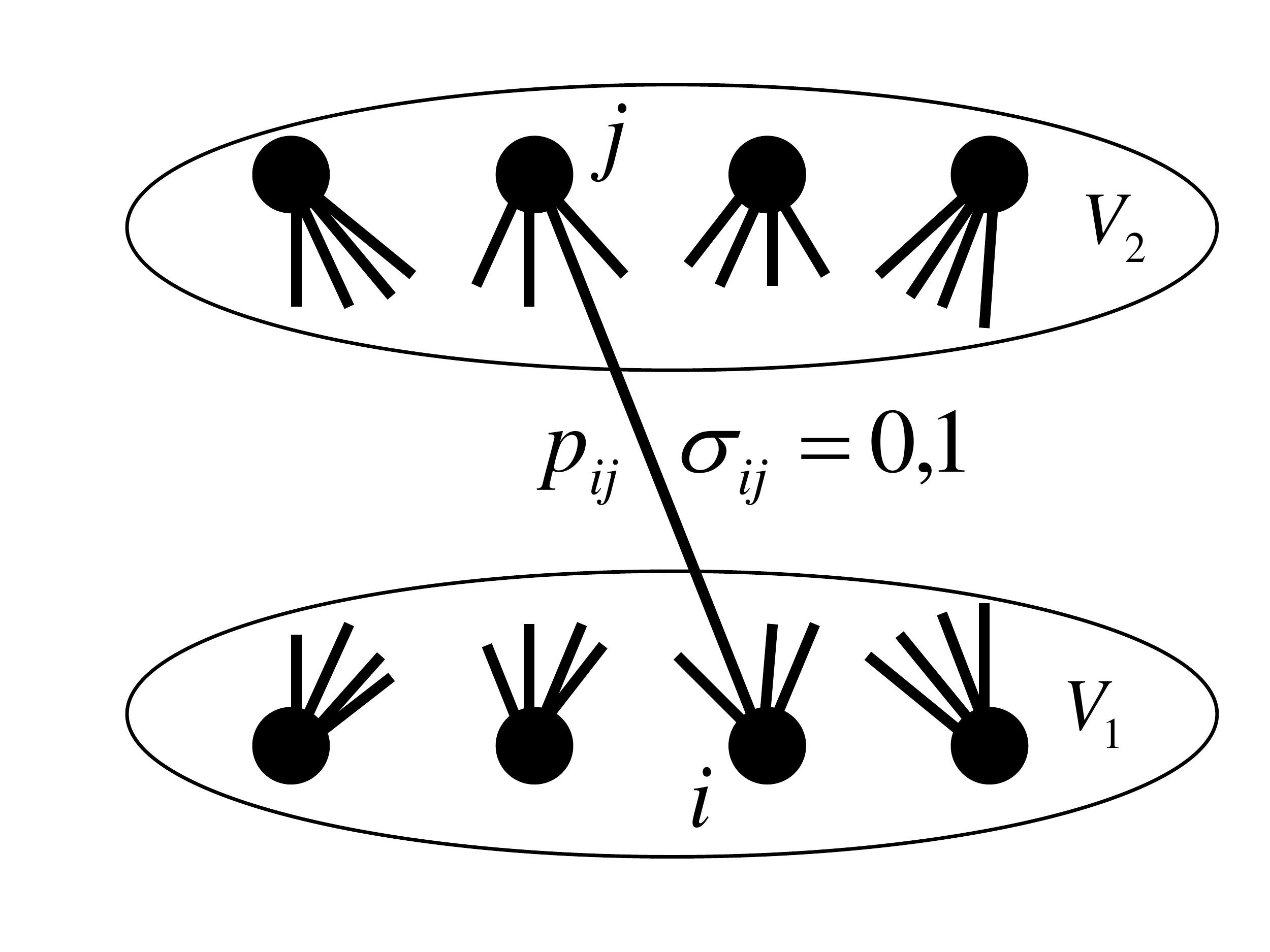}
\caption{Illustration of Graphical Model for perfect matchings and permanent.}
\label{fig:GM}
\end{figure}

The permanent of a matrix can be interpreted as the partition function $Z$ of a Graphical Model (GM) defined over a bipartite undirected graph, ${\cal G}=({\cal V}=({\cal V}_1,{\cal V}_2),{\cal E})$, where ${\cal V}_1$, ${\cal V}_2$ are of equal size, $|{\cal V}_1|=|{\cal V}_2|=n$, and ${\cal V}_1$, ${\cal V}_2$, and ${\cal E}$ stand for the set of $n$ vertices/labels of particles in the first and second images and the set of edges (possible relations) between particles in the two images, respectively. The basic binary variables, $\sigma_{ij}=0,1$, are associated with the edges, while the perfect matchings are enforced via the constraints associated with vertexes, $\forall i\in{\cal V}_1:\quad
\sum_{j\in V_2}\sigma_{ij}=1$ and $\forall j\in{\cal V}_2:\quad
\sum_{i\in V_1}\sigma_{ij}=1$, as illustrated in Fig.~\ref{fig:GM}. A non-negative element of the matrix, $p_{ij}$, turns into the weight associated with the edge $(i,j)$. In summary, the GM relates the following probability to any of $n!$ perfect matchings, $\sigma$:
\begin{align}
& \varrho(\sigma)=Z(p)^{-1}\prod_{(i,j)\in{\cal E}} (p_{ij})^{\sigma_{ij}}, \label{GM}\\
& \sigma=\Biggl(\sigma_{ij}=0,1\left|(i,j)\in{\cal E};\ \forall i\in{\cal V}_1:\  \sum_{j\in{\cal V}_2}\sigma_{ij}=1;\forall j\in{\cal V}_2:\  \sum_{i\in{\cal V}_1}\sigma_{ij}=1 \right.\Biggr),\nonumber
\\
& Z(p)=\perm(p)=\sum_\sigma \prod_{(i,j)\in{\cal E}} (p_{ij})^{\sigma_{ij}}.
\label{Z-def}
\end{align}

The GM formulation (\ref{GM}) also suggests a variational, Kullback-Leibler (KL) scheme for computing the permanent. The only minimum of the so-called exact Free Energy (FE) functional,
\begin{equation}
F(b|p)=\sum_\sigma b(\sigma)\log\left(\frac{b(\sigma)}{\prod_{(i,j)\in{\cal E}} (p_{ij})^{\sigma_{ij}}}\right),
\label{KL}
\end{equation}
computed over $b(\sigma) \geq 0$ for all $\sigma$ under the normalization condition,  $\sum_\sigma b(\sigma)=1$, is achieved at $b(\sigma)=\varrho(\sigma)$, and the value of the exact FE functional at the minimum over $b(\sigma)$ is $-\log(Z(p))$. Here,  the general and the optimal $b(\sigma)$ are interpreted as, respectively, the proxy and the probability of the perfect matching $\sigma$.

The relation between the problem of computing the permanent and the problem of finding the most probable (maximum) perfect matching is discussed in Appendix \ref{subsec:ML}.

\subsection{Exact Methods for Computing Permanents}

Computing the permanent of a matrix is a $\#$-P hard problem,  i.e., it is a problem which most likely requires a number of operations exponential in the size of the matrix. In Appendix \ref{sec:Ryser_vs_ZDD}, we experiment and compare the performance of the following two exact deterministic ways to evaluate permanents:
\begin{itemize}
\item A general method based on Zero-supressed binary Decision Diagrams (ZDDs), explained in more detail in \citet{Knuth:2009:ACP:v4f1}. See also detailed explanations below in Appendix \ref{sec:ZDD}. As argued in \citet{Knuth:2009:ACP:v4f1}, the ZDDs may be a very efficient practical tool for computing partition functions in general graphical models. This thesis was illustrated by \citet{09YedA} on the example of counting independent sets and kernels of graphs.

\item A permanent-specific method based on Ryser's formula:
$$
Z(p)=(-1)^n
\sum_{S\subseteq \{1,\ldots,n\}}(-1)^{|S|}\prod_{i=1}^n\sum_{j\in S} p_{ij}.
$$
We use code from \citet{Ryser_code} implementing the Ryser's formula.
\end{itemize}

Note that in most practical cases many entries of $p$
are very small and they do not affect the permanent of $p$ significantly. These entries do, however, take computational resources if accounted for in the algorithm. To make computations efficient we sparsify the resulting matrix $p$, implementing the heuristic pruning technique explained in Appendix \ref{sec:pruning}.

We also verify some of our results against randomized computations of the permanent using the FPRAS from \citet{04JSV},  with a specific implementation from \citet{08CKV}.

\section{Approximate Methods and Exact Relations}
\label{sec:BP+MF}

We perform an approximate computation of the permanent by following the general BFE approach of \citet{05YFW} and the associated Belief Propagation/Bethe-Peierls (BP) algorithm,  discussed in detail for the case of permanents of positive matrices in \citet{08CKV,09HJ,10CKKVZ}. (See also Appendix \ref{sec:BP} reproducing the description of \citet{08CKV,10CKKVZ} and presented in this manuscript for convenience.) In our BP experiments we implement the algorithm discussed by \citet{08CKV} with a special type of initialization corresponding to the best perfect matching of $p$. We also generalize the BP scheme by modifying the entropy term in the BFE.

In the following subsections we re-introduce the BFE approach, the related but different Mean Field FE approach,  and also consider a fractional FE approach generalizing and interpolating between Bethe/BP and MF approaches.  Even though these optimization approaches and respective algorithms can be thought of as approximating the permanent  we will show that they also generate some exact relations for the permanent.

\subsection{Belief Propagation/Bethe-Peierls Approach}
\label{subsec:BP}

Let us start by defining some useful notation.
\begin{definition}[$\beta$-polytope]
\label{def:beta-polytope}
Call the $\beta$-polytope of the non-negative matrix $p$ (or just $\beta$-polytope for short) the set of doubly stochastic non-negative matrices with elements corresponding to zero elements of $p$ equal to zero, ${\cal B}_p=
(\beta_{ij}|\forall i:\ \sum_{(i,j)\in\cal E}\beta_{ij}=1;
\forall j:\ \sum_{(i,j)\in\cal E}\beta_{ij}=1; \forall (i,j) \mbox{ with } p_{ij}=0\ \beta_{ij}=0 \mbox{ holds})$. We say that $\beta$ lies in the interior of the $\beta$-polytope, $\beta\in {\cal B}_p^{(int)}$, if $\forall (i,j)\mbox{ with }p_{ij}\neq 0\ \beta_{ij}\neq 0,1\mbox{ holds}$.
\footnote{Let us mention, that what we call here ``interior" would be mathematically more accurate to call ``relative interior," see e.g. \cite{04BV}.}
\end{definition}
In English,  the interior solution means that all elements of the doubly stochastic $\beta$ are non-integer, under exception of the case when $p_{ij}=0$ and, respectively, $\beta_{ij}=0$.

\begin{definition}[Bethe Free Energy (for the permanent)]
\label{def:BFE}
The following functional of $\beta\in {\cal B}_p$
\begin{eqnarray}
F_{BP}(\beta|p)=\sum_{(i,j)}\left(\beta_{ij}\log(\beta_{ij}/p_{ij})\!-\!
(1\!-\!\beta_{ij})\log(1\!-\!\beta_{ij})\right),
\label{F_BP}
\end{eqnarray}
conditioned to a given $p$, is called the Bethe Free Energy (BFE) or the Belief-Propagation/Bethe-Peierls (BP) functional (for the permanent)
\footnote{In the following,  and whenever Bethe, MF, or fractional FE are mentioned, we will drop the clarifying -- for the permanent -- as only permanents are discussed in this manuscript.}
\end{definition}

To motivate the definition above let us briefly discuss the concept of the Bethe FE which was introduced in \citet{05YFW} for the case of a general pair-wise interaction GM (with variables associated with vertices of the graph). Schematically, the logic extended to the case with variables associated with edges of the graph and leading to Eq.~(\ref{F_BP}) for the permanent is as follows.  (See \citet{10WC} for a detailed discussion.) Consider a GM with binary variables associated with edges of the graph. If the graph is a tree,  then the following exact relation holds, $\rho(\sigma)=\prod_i \rho_i(\sigma_i)/\prod_{(i,j)} \rho_{ij}(\sigma_{ij})$, where $\sigma_i=(\sigma_{ij}=0,1|(i,j)\in{\cal E})$. Here, $\rho_i(\sigma_i)$ and $\rho_{ij}(\sigma_{ij})$ are marginal probabilities associated with vertex $i$ and edge $(i,j)$ of the graph. Replacing the probabilities by their proxies/beliefs,
$\rho(\sigma)\to b(\sigma)$, $\rho_i(\sigma_i)\to b_i(\sigma_i)$ and $\rho_{ij}(\sigma_{ij})\to b_{ij}(\sigma_{ij})$, substituting the ratio of probabilities expression for $b(\sigma)$ in the exact FE functional (\ref{KL}), and accounting for relations between the marginal beliefs, one arrives at the general expression for the Bethe FE functional.  This expression for the Bethe FE functional is exact on a tree only,  and it is similar in spirit to the one introduced in \citet{05YFW} as an approximation for GM on a graph with loops. When the graph is bi-partite with the equal number of nodes in the two parts the BP replacement for $b(\sigma)$ becomes
\begin{equation}
b_{BP}(\sigma)=\frac{\prod_{(i,j)\in{\cal E}:\sigma_{ij}=1}\beta_{ij}}{\prod_{(i,j)\in{\cal E}:\sigma_{ij}=0}(1-\beta_{ij})},
\label{b-bp}
\end{equation}
where $\beta_{ij}=b_{ij}(1)$ is the marginal belief correspondent to finding the edge $(i,j)$ in the matching.
Then, substituting $b(\sigma)$ by $b_{BP}(\sigma)$ in Eq.~(\ref{KL}) results in the Bethe FE expression (\ref{F_BP}) for the perfect matchings (permanents). Note, that while the exact FE (\ref{KL}) is the sum of
$O(n!)$ terms,  there are only $O(n^2)$ terms in the Bethe FE (\ref{F_BP}).

According to the Loop Calculus approach of \citet{06CCa,06CCb}, extended to the case of the permanent in \citet{08CKV,10CKKVZ,10WC}, the BP expression and the permanent are related to each other as follows:
\begin{theorem}[Permanent and BP, \citep{10WC}]
\label{Prop:Perm+BP}
If the BP equations following from minimization of the the BFE (\ref{F_BP}) over the doubly stochastic $\beta$,
\begin{eqnarray}
\forall (i,j):\quad (1-\beta_{ij})\beta_{ij}=\frac{p_{ij}}{u_i u^j},\label{BP2}
\end{eqnarray}
where $u_i$ and $u^j$ are positive-valued Lagrangian multipliers correspondent to the $\sum_{j\in {\cal V}_2} \beta_{ij}=1$ and $\sum_{i\in {\cal V}_1} \beta_{ij}=1$ constraints respectively, have a solution in the interior of the $\beta$-polytope, $\beta\in {\cal B}_p^{(int)}$, then
\begin{eqnarray}
Z=Z_{BP}(p)\perm(\beta.*(1-\beta))\frac{1}{\prod_{i,j}(1-\beta_{ij})},
\label{LC}
\end{eqnarray}
where $Z_{BP}(p)=-\log(F_{BP}(\beta|p))$.
\end{theorem}
The proof of the Theorem \ref{Prop:Perm+BP} also appears in Appendix \ref{subsec:BP-exact}.
An iterative heuristic algorithm solving BP Eqs.~(\ref{BP2}) for the doubly stochastic $\beta$ efficiently is discussed in Appendix \ref{sec:BP-iterat}.

Let us recall that the $(i,j)$ element of the doubly stochastic $\beta$, $\beta_{ij}$, is interpreted as the proxy (approximation) to the marginal probability for the $(i,j)$ edge of the bipartite graph ${\cal G}$ to be in a perfect matching, i.e., $\beta_{ij}$, should be thought of as an approximation for $\varrho_{ij}=\sum_{\sigma:\, \sigma_{ij}=1}\varrho(\sigma)$.

Note also that $\log(u_i)$ and $\log(u^j)$ in Eqs.~(\ref{BP2}) are the Lagrange multipliers related to the $2n$ double stochasticity (equality) constraints on $\beta$.

\subsubsection{BP as the Minimum of the Bethe Free Energy}
\label{subsubsec:o-BP}

\begin{definition}[Optimal Bethe Free Energy] We define optimal BFE, $F_{o-BP}(p)$, and resepctive counting factor,
$Z_{o-BP}(p)$, according to
\begin{eqnarray}
-\log(Z_{o-BP}(p))=F_{o-BP}(p)=\min_{\beta\in {\cal B}} F_{BP}(\beta|p),
\label{o-BP}
\end{eqnarray}
where $F_{BP}(\beta|p)$ is defined in Eq.~(\ref{F_BP}).
\end{definition}

Considered in the general spirit of \citet{05YFW},
$F_{o-BP}(p)$, just defined, should be understood as an approximation to $-\log(\perm(p))$.
To derive Eq.~(\ref{o-BP}) one needs to replace $b(\sigma)$ by its lower parametric proxy (\ref{b-bp}). (See \citet{10WC} for more details.)

The relation between the optimization formulation (\ref{o-BP}) and the BP Eqs.~(\ref{BP2}) requires some clarifications stated below in terms of the following two propositions.

\begin{proposition}[Partially Resolved BP Solutions]
\label{prop:part-BP}
Any doubly stochastic $\beta$ solving Eqs. (\ref{BP2}) and lying on the boundary of the $\beta$ polytope, i.e., $\beta\in {\cal B}_p$ but $\beta\notin {\cal B}_p^{(int)}$, can be reduced by permutations of rows and columns of $\beta$ (and $p$, respectively) to a block diagonal matrix,  with one block consisting of $0,1$ elements only and corresponding to a partial perfect matching,  and the other block having all elements strictly smaller than unity,  and nonzero if the respective $p_{ij}\neq 0$. We call such a solution of the BP Eqs.~(\ref{BP2}) partially resolved solutions, emphasizing that a part of the solution forms a partial perfect matching, and any other perfect matching over this subset is excluded by the solution (in view of the probabilistic interpretation of $\beta$). A doubly stochastic $\beta$ corresponding to a full perfect matching is called a fully resolved solution of the BP Eqs.~(\ref{BP2}).
\end{proposition}
\begin{proof}
This statement follows directly from the double stochasticity of $\beta$ and from the form of the BP Eqs.~(\ref{BP2}),  and it was already discussed in \citet{08CKV,10WC} for the fully resolved case.
\end{proof}

\begin{proposition}[Optimal Bethe FE and BP equations]
The optimal Bethe FE, $F_{o-BP}(\beta)$ over $\beta\in {\cal B}_p$, can only be achieved at a solution of the BP Eqs.~(\ref{BP2}), possibly with the Lagrange multipliers $u_i$, $u^j$ taking the value $+\infty$.
\end{proposition}
\begin{proof}
This statement is an immediate consequence of the fact that Proposition \ref{prop:part-BP} is valid for any $p$, and so a continuous change in $p$ (capable of covering all possible achievable $p$) can only result in an interior solution for the doubly stochastic $\beta$ merging into a vertex of the $\beta$-polytope, or emerging from the vertex (than respective Lagrangian multipliers take the the value $+\infty$),  but never reaching an edge of the polytope at any other location but a vertex. Therefore, we can exclude the possibility of achieving the minimum of the Bethe FE anywhere but at an interior solution, partially resolved solution or a fully resolved solution (corresponding to a perfect matching) of the BP equations.
\end{proof}

Note,  that an example where the minimum in Eq.~(\ref{o-BP}) is achieved at the boundary of the $\beta-polytope$ (in fact, at the most probable perfect matching corner of the polytope) was discussed in \citet{10WC}. 

Another useful and related statement, made recently in \citet{11Von}, is
\begin{proposition}[Convexity of the Bethe FE, \cite{11Von}]
\label{prop:BP-convexity}
The Bethe FE (\ref{F_BP}) is a convex functional of $\beta\in {\cal B}_p$.
\end{proposition}
A few remarks are in order. First, the statement above is nontrivial as, considered naively, individual edge contributions in Eq.~(\ref{F_BP}) associated with the entropy term, $\beta_{ij}\log\beta_{ij}-(1-\beta_{ij})\log(1-\beta_{ij})$, are not convex for $\beta_{ij}>1/2$,  and the convexity is restored only due to the global (double stochasticity) condition. Second, the convexity means that if the optimal solution is not achieved at the boundary of ${\cal B}_p$,
then either the solution is unique (general case) or the situation is degenerate and one finds a continuous family of solutions all giving the same value of the Bethe FE. The degeneracy means that $p$ should be fine tuned to get into the situation, and addition of an almost any small (random) perturbation to $p$ would remove the degeneracy.  To illustrate how the degeneracy may occur, consider an example of a $(2\times 2)$ matrix $p$ with all elements equal to each other. We first observe that regardless of $p$ for $n=2$, the entropy contributions to the Bethe FE are identical to zero for any doubly stochastic $(2\times 2)$ matrix, $\sum_{(i,j)}^{i,j=1,2}(\beta_{ij}\log \beta_{ij}-(1-\beta_{ij})\log(1-\beta_{ij}))=0$. Moreover, the remaining, linear in $\beta$, contribution to the Bethe FE (which is also called the self-energy in physics) turns into a constant for the special choice of $p$. Thus one finds that in this degenerate $n=2$ case,
$$
\beta=\left(\begin{array}{cc} \alpha & 1-\alpha\\ 1-\alpha & \alpha \end{array}\right),
$$
with any $\alpha\in [0;1]$, is a solution of Eqs.~(\ref{BP2}) also achieving the minimum of the Bethe FE. Creating any asymmetry between the four components of the $(2\times 2)$ $p$ will remove the degeneracy,  moving the solution of Eqs.~(\ref{BP2}) achieving the minimum of the Bethe FE to one of the two perfect matchings, correspondent to $\alpha=0$ and $\alpha=1$, respectively. It is clear that this special ``double" degeneracy (first, cancellation of the entropy contribution,  and then
constancy of the self-energy term) will not appear at all if the doubly stochastic $\beta$, solving Eqs.~(\ref{BP2}) in the $n>2$ case has more than two nonzero components in every row and column. Combined with Proposition \ref{prop:BP-convexity}, this observation translates into the following statement.
\begin{corollary}[Uniqueness of interior BP solution]
\label{cor:BP-uniqness}
If $n>2$ and an interior, $\beta\in{\cal B}_p^{(int)}$, solution of Eqs.~(\ref{BP2}) has more than two nonzero elements in every row and column, then the solution is unique \footnote{In the following, discussing an interior BP solution, $\beta\in{\cal B}_p^{(int)}$ and aiming to focus only on the interesting/nontrivial cases, we will  be assuming that $n>2$ and $p$ has more than two nonzero elements in every row and column.}.
\end{corollary}

\subsection{Mean-Field Approach}
\label{subsec:MF}

\begin{definition}[Mean Field Free Energy]
For $\beta\in {\cal B}_p$, the MF FE is defined as
\begin{eqnarray}
F_{MF}(\beta|p)=\sum_{(i,j)}\left(\beta_{ij}\log(\beta_{ij}/p_{ij})+(1-\beta_{ij})\log(1-\beta_{ij})\right).
\label{MF1}
\end{eqnarray}
\end{definition}

Let us precede discussion of usefulness of the MF notion/approach by a historical and also motivational remark. MF is normally thought of as an approximation ignoring correlations between variables.  Then, the joint distribution function of $\sigma$ is expressed in terms of the product of marginal distributions of its components. In our case of the perfect matching GM over the bi-partite graph, the MF approximation constitutes the following substitution for the exact beliefs,
\begin{equation}
b(\sigma)\to \prod_{(i,j)\in{\cal E}}b_{ij}(\sigma_{ij}),
\label{MF-subs}
\end{equation}
into Eq.~(\ref{KL}). Making the substitution and relating the marginal edge beliefs to $\beta$ according to, $\forall (i,j)\in{\cal E}:\ b_{ij}(1)=\beta_{ij},\ b_{ij}(0)=1-\beta_{ij}$, one arrives at Eq.~(\ref{MF1}). Because of how the perfect matching problem is defined, the two states of an individual variable,
$\sigma_{ij}=0$ and $\sigma_{ij}=1$,  are in the exclusion relation, and so one can also associate the special form of Eq.~(\ref{MF1}) with the exclusion or Fermi- (for Fermi-statistics of physics) principle.

Direct examination of Eq.~(\ref{MF1}) reveals that
\begin{proposition}[MF FE minimum is always in the interior]
\label{prop:MF}
$F_{MF}(\beta|p)$ is strictly convex and its minimum is achieved at $\beta\in {\cal B}_p^{(int)}$.
\end{proposition}

Looking for the interior minimum of Eq.~(\ref{MF1}) one arrives at the following MF equations for the (only) stationary point of the MF FE functional
\begin{eqnarray}
&& \forall (i,j)\in{\cal E}:\quad\beta_{ij}=\frac{1}{1+v_i v^j/p_{ij}},
\label{MF2}
\end{eqnarray}
where $v_i$ and $v^j$ are Lagrangian multipliers enforcing the conditions, $\sum_j\beta_{ij}=1$ and $\sum_i\beta_{ij}=1$, respectively. The equations can also be rewritten as
\begin{eqnarray}
\forall (i,j):\quad \frac{\beta_{ij}}{1-\beta_{ij}}=\frac{p_{ij}}{v_i v^j},\label{MF3}
\end{eqnarray}
making comparison with the respective BP Eqs.~(\ref{BP2}) transparent. An efficient iterative algorithm for solving the MF equations (\ref{MF3}) is discussed in Appendix \ref{subsec:MF-iterat}.

Direct examination shows that (unlike in the BP case)  $\beta$ with a single element equal to unity or zero (when the respective $p$ element is nonzero) cannot be a solution of the MF Eqs.~(\ref{MF3}) over doubly stochastic $\beta$ -- fully consistently with the Proposition \ref{prop:MF} above. Moreover, $Z_{o-MF}(p)$,  defined as  $-\log$ of the minimum of the MF FE (\ref{MF1}), is simply equal to $Z_{MF}(p)$, defined as $-\log F_{MF}(\beta)$ evaluated at the (only) doubly stochastic solution of Eq.~(\ref{MF3}).

Note also that the MF functional (\ref{MF1}) cannot be considered as a variational proxy for the permanent,  bounding its value from below. This is because the substitution on the right-hand side of Eq.~(\ref{MF-subs}) does not respect the perfect matching constraints, assumed reinforced on the left-hand side of Eq.~(\ref{MF-subs}). In particular, the probability distribution function on the right-hand side of Eq.~(\ref{MF-subs}) allows two edges of the graph adjacent to the same vertex to be in the active, $\sigma_{ij}=1$, state simultaneously. However, this state is obviously prohibited by the original probability distribution,  on the left-hand side of Eq.~(\ref{MF-subs}) defined only over $n!$ of states corresponding to the perfect matchings.  As shown below in Section \ref{subsec:fractional-opt}, the fact that the MF ignores the perfect matching constraints results in the estimation $Z_{MF}(p)$ upper bounding $\perm (p)$, contrary to what a standard MF (not violating any original constraints) would do.

Finally and most importantly (for the MF discussion of this manuscript), the MF approximation for the permanent, $Z_{MF}$,  can be related to the permanent itself as follows:
\begin{theorem}[Permanent and MF]
\begin{eqnarray}
Z(p)=\perm(p)=Z_{MF}(p)\perm\left(\beta./(1-\beta)\right)\prod_{(i,j)\in{\cal E}}(1-\beta_{ij}),
\label{MF-exact}
\end{eqnarray}
where $\beta$ is the only interior minimum of (\ref{MF1}).
\end{theorem}
The proof of this statement is given in Appendix \ref{subsec:MF-exact}.

\subsection{Fractional Approach}
\label{subsec:frac}

Similarity between the exact BP expression (\ref{LC}) and the exact MF expression (\ref{MF-exact}) suggests that the two formulas are the limiting instances of a more general relation. Indeed,  one finds that
\begin{theorem}[Fractional representation for permanent]
\label{theorem:m-exact}
For any non-negative $p$ and doubly stochastic $\beta$ which solves
\begin{eqnarray}
\forall (i,j):\quad \frac{\beta_{ij}}{(1-\beta_{ij})^\gamma}=\frac{p_{ij}}{w_i w^j},\label{m3}
\end{eqnarray}
for $\gamma\in [-1;1]$,
and if the solution found is in the interior of the domain, i.e., $\beta\in {\cal B}_p^{(int)}$,
the following relation holds
\begin{eqnarray}
\perm(p)=Z_f^{(\gamma)}(\beta|p) \perm\left(\frac{\beta.}{(1-\beta).^\gamma}\right)\prod_{(i,j)}(1-\beta_{ij})^\gamma,
\label{m-exact}
\end{eqnarray}
where
\begin{eqnarray}
F_f^{(\gamma)}(\beta|p)=-\log(Z_f^{(\gamma)}(\beta|p))=
\sum_{(i,j)}\left(\beta_{ij}\log(\beta_{ij}/p_{ij})+
\gamma(1-\beta_{ij})\log(1-\beta_{ij})\right);
\label{m1}
\end{eqnarray}
and $w_i$ and $w^j$ in Eq.~(\ref{m3}) are the Lagrangian multipliers enforcing the conditions, $\sum_j\beta_{ij}=1$ and $\sum_i\beta_{ij}=1$, respectively.
\end{theorem}
The proof of Eq.~(\ref{m1}) is given in Appendix \ref{subsec:fractional-exact}.
An iterative heuristic algorithm solving Eqs.~(\ref{m3}) efficiently is described in Appendix \ref{subsec:fractional-iterat}.

Following the general GM logic and terminology introduced in \citet{Wiegerinck_Heskes_2003,05YFW},
we call $F_f^{(\gamma)}(\beta|p)$ the fractional FE. Obviously the two extremes of $\gamma=-1$ and $\gamma=1$ correspond to BP and MF limits, respectively. Many features of the BP and MF approaches extend naturally to the fractional case.  In particular,  one arrives at the following statement.
\begin{proposition}[Fractional Convexity, Theorem 60 of \citet{11Von}]
\label{prop:fract-cov}
The fractional functional defined in Eq.~(\ref{m1}), $F_f^{(\gamma)}(\beta|p)$,
is a convex functional, convex over $\beta\in{\cal B}_p$ for any $\gamma\in[-1;1]$ and any non-negative $p$.
\end{proposition}
Obviously,  this statement generalizes Proposition \ref{prop:BP-convexity}. Also, the following statement becomes a direct consequence of Proposition \ref{prop:fract-cov}:
\begin{corollary}[Uniqueness of the interior fractional minimum]
If the minimum of $F_f^{(\gamma)}(\beta|p)$ is realized at $\beta\in {\cal B}_p^{(int)}$, it is unique.
\end{corollary}

\subsection{Minimal Fractional Solution}
\label{subsubsec:opt-mix}

It is clear that at $\gamma>0$ the fractional Eqs.~(\ref{m3}) cannot have a perfect matching solution, thus suggesting that at least in this case the solution, if exists, is in the interior, $\beta\in {\cal B}_p^{(int)}$. On the other hand general existence (for any $p$) of such a solution follows immediately from the existence in a special case, for example of $p$ with all elements equal,  and then from the continuity of the Eqs.~(\ref{m3}) solution with respect to $p$.

The case of $\gamma \in ]-1;0]$ is a bit trickier.  In this case, Eqs.~(\ref{m3}) formally do allow a perfect matching solution.
However, for all but degenerate $p$,  i.e. one reducible by permutations to a diagonal matrix, the perfect matching solution is an isolated point. Indeed,  let us consider a vicinity of a degenerate $p$. If one picks (without loss of generality) a diagonal, $p^{(0)}=(a_i\delta_{ij}| (i,j)\in{\cal E})$, and consider $p=p^{(0)}+\delta$,  where $\delta$ is a small positive matrix, then one observes that Eqs.~(\ref{m3}) do allow a solution, $\beta=1+\epsilon b$,  where $\epsilon$ is a small positive scalar
and $b=(b_{ij}|(i,j)\in{\cal E};
\forall i\in {\cal V}_1:\ \sum_j b_{ij}=0;\ \forall j\in {\cal V}_2:\ \sum_i b_{ij}=0)$ is a matrix with $O(1)$ elements, if the following scaling relation holds, $|\delta|\sim \epsilon^{1+\gamma}$. Moreover, one also finds that a solution $\beta$ is $\epsilon$-close to a perfect matching only if $p$ is $\epsilon^{1+\gamma}$-close to a diagonal matrix. Now we apply the same continuity and existence arguments, as used above in the $\gamma>0$ case arguments, to find out that the following statement holds.
\begin{proposition}[Fractional Minima]
\label{prop:frac_min}
The minimal fractional solution, defined by
\begin{eqnarray}
-\log(Z_{o-f}^{(\gamma)}(p))=F_{o-f}^{(\gamma)}(p)=
\min_{\beta\in {\cal B}_p} F_{f}^{(\gamma)}(\beta|p),
\label{o-m}
\end{eqnarray}
can only be achieved for $\gamma>-1$ and general (non-degenerate) $p$ at $\beta\in {\cal B}_p^{(int)}$.
\end{proposition}

Then the following statement follows.
\begin{proposition}[$\gamma$-monotonicity and continuity]
\label{prop:gamma_mon}
For any non-negative $p$, $Z_{o-f}^{\gamma}(p)$ is a monotonically increasing and continuous function of $\gamma$ in $[-1;1]$.
\end{proposition}
\begin{proof}
Observe that for any non-negative $p$ and doubly stochastic $\beta$, $\sum_{(i,j)\in{\cal E}}(1-\beta_{ij})\log(1-\beta_{ij})<0$, so for any $\gamma_{1,2}\in[-1;1]$ such that $\gamma_1>\gamma_2$,  $F_f^{(\gamma_1)}(\beta|p)\leq F_f^{(\gamma_2)}(\beta|p)$. Then according to the definition of $F_{o-f}^{(\gamma)}(p)$, $F_{o-f}^{(\gamma_1)}(p)\leq F_f^{(\gamma_1)}(\beta|p)\leq F_f^{(\gamma_2)}(\beta|p)$, for any doubly stochastic $\beta$, in particular for $\beta$ which is optimal for $\gamma_2$. Finally, $F_{o-f}^{(\gamma_1)}(p)\leq F_{o-f}^{(\gamma_2)}(p)$, proving monotonicity. The continuity of $Z_{o-f}^{(\gamma)}(p)$ with respect to $\gamma$ in $]-1;1]$ follows from the Proposition \ref{prop:frac_min} combined with  $F_f^{(\gamma)}(\beta|p)$ continuity with respect to both $\gamma\in[-1;1]$ and $\beta\in{\cal B}_p^{(int)}$. (The intuition with respect to the continuity is as follows: an increase in $\gamma$ pushes the optimal $\beta$ away from the boundary of the ${\cal B}_p$ polytope.)
\end{proof}

\section{Permanent Inequalities, Special Value of $\gamma$, and Conjectures}
\label{sec:what_to_test}

We start this section discussing in Subsection \ref{subsec:ineq} the recently derived permanent inequalities related to BP and MF analysis.  Then, we switch to describing new results of this manuscript in Subsection \ref{subsec:fractional-opt}, which are mainly related  to the fractional generalizations of the inequalities discussed in Subsection \ref{subsec:ineq}. We also discuss in Subsection \ref{subsec:fractional-opt} the special (and $p$-dependent) value of the fractional coefficient $\gamma$ for which $\perm(p)$ is equal to $Z_{o-f}^{\gamma}(p)$. Finally, Subsection \ref{subsec:conject} is devoted to discussing conjectures whose resolutions should help to tighten bounds for the permanent.

\subsection{Recently Derived Inequalities}
\label{subsec:ineq}

In this subsection we discuss a number of upper and lower bounds on permanents of positive matrices introduced recently. Our task is two-fold. First,  we wish to relate the bounds/inequalities to the BP and MF approaches introduced and discussed in the preceding section. Some of these relations and interpretations are new. However, we also aim to test these bounds,  and specifically to characterize the tightness of the bounds by testing the gap as a function of advection and diffusion parameters in the 2d diffusion+advection model in Section \ref{sec:Experiments}.

The first bound of interest is
\begin{proposition}[BP lower bound]
For any non-negative $p$
\begin{eqnarray}
\perm(p)\geq Z_{o-BP}(p). \label{B1}
\end{eqnarray}
\end{proposition}
This statement, as an experimental but unproven observation, was made in \citet{08CKV}. It was stated as a theorem (Theorem \# 14) in \citet{10Von}, but the proof was not provided. (See also discussion in \citet{11Von} following Theorem 49/Corollary 50.) The statement was proven in \citet{11Gur}. Interpreted in terms of the terminology and logic introduced in the preceding Sections, the proof of \citet{11Gur} consists (roughly) in combining the inequality  by \citep{98Sch}
\begin{eqnarray}
\perm(\beta.*(1-\beta))\geq \prod_{(i,j)} (1-\beta_{ij}),
\label{Sch-ineq}
\end{eqnarray}
stated for any doubly stochastic $\beta$,
with some (gauge) manipulations/transformations of the type discussed above in Sections \ref{subsubsec:o-BP}. We give our version of the proof (similar to the one in \citet{11Gur} in spirit, but somewhat different in details) in Appendix \ref{sec:BP_low}. One direct corollary of the bound (\ref{B1}) discussed in \citet{11Gur}, is that
\begin{corollary}
\label{cor:Gurvits}
For an arbitrary doubly stochastic $\phi$
\begin{eqnarray}
\perm(\phi)\geq Z_{o-BP}(\phi)\geq \prod_{(i,j)} (1-\phi_{ij})^{1-\phi_{ij}}. \label{B1-ext}
\end{eqnarray}
\end{corollary}

Next, the following two lower bounds follow from analysis of Eq.~(\ref{LC}).
\begin{proposition}[BP lower bound \#1]
\label{prop:bound3}
For any non-negative $p$ and doubly stochastic $\beta\in {\cal B}_p^{(int)}$ solving Eqs.~(\ref{BP2}) (if the solution exists) results in
$$
\perm(p)\geq Z_{BP}(\beta|p) \prod_{(i,j)}(1-\beta_{ij})^{\beta_{ij}-1}\frac{n!}{n^n}.
$$
\end{proposition}
This is the statement of Corollary 7 of \citet{10WC} valid for any interior point solution of the BP-equations,  and it follows from the Gurvits-van der Waerden theorem of \citet{08Gur,09LS}, also stated as Theorem 6 in \citet{10WC}.

\begin{proposition}[BP lower bound \#2]
\label{prop:bound4}
For any non-negative $p$ and $\beta\in {\cal B}_p^{(int)}$ solving Eqs.~(\ref{BP2}) (if the solution exists) results in
$$
Z\geq 2 Z_{BP}(\beta|p)(\prod_{i,j}(1-\beta_{ij}))^{-1}\prod_i \beta_{i\Pi(i)}(1-\beta_{i\Pi(i)}),
$$
where $\Pi$ is an arbitrary permutation.
\end{proposition}
This statement was made in Theorem 8 in \citet{10WC} and it is also related to an earlier observation of \citet{73ES}
\footnote{The proof of the Theorem 8 in \citet{10WC} contained a misprint that was corrected in the erratum available at
\url{https://sites.google.com/site/mchertkov/publications/mypapers/91_erratum.pdf}.}

\begin{proposition}[BP upper bound]
\label{prop:bound5}
For any non-negative $p$ and $\beta\in {\cal B}_p^{(int)}$ solving Eqs.~(\ref{BP2}) (if the solution exists)
\begin{eqnarray}
\perm(p)\leq Z_{BP}(\beta|p)(\prod_{(i,j)\in{\cal E}}(1-\beta_{ij}))^{-1}\prod_j(1-\sum_i(\beta_{ij})^2),
\label{bound5}
\end{eqnarray}
\end{proposition}
This statement was made in Proposition 9 of \citet{10WC}.

\subsection{New Bounds and $\gamma_*$}
\label{subsec:fractional-opt}

Of the bounds discussed above, three are related to BP and one to MF,
while as argued in Section
\ref{subsec:conject} the fractional approach interpolates between BP and MF. This motivates exploring below new fractional generalizations of the previously known (and discussed in the preceding subsection) BP and MF bounds.

We first derive a new lower bound generalizing Proposition \ref{prop:bound3} to the fractional case.
\begin{proposition}
\label{prop:fractional-low}
The following is true for any doubly stochastic $\beta$ and any $\gamma\in[-1;1]$
$$
\perm\left(\beta.*(1-\beta).^{-\gamma}\right)\geq \frac{n!}{n^n}\prod_{(i,j)}(1-\beta_{ij})^{-\gamma\beta_{ij}}.
$$
\end{proposition}
\begin{proof}
This bound generalizes Corollary 7 of \citet{10WC}, and it follows directly from the Gurvits--van der Waerden theorem of \citet{08Gur,09LS} (see also Proposition 8 of \citet{10WC},  where a misprint should be corrected $n^n/n!\to n!/n^n$), and the inequality,
$\sum_j\beta_{ij}(1-\beta_{ij})^{-\gamma}x_j\geq \prod_j\left(\left(1-\beta_{ij}\right)^{-\gamma}x_j\right)^{\beta_{ij}}$.
\end{proof}
Then, combining Proposition \ref{prop:fractional-low} with Theorem \ref{theorem:m-exact}, one arrives at the following statement, generalizing Proposition \ref{prop:bound3}:
\begin{corollary}[fractional low bound]
\label{corr:fractional-low}
For any non-negative $p$ and $\beta\in {\cal B}_p^{(int)}$ solving Eqs.~(\ref{m3}) (if the solution exists), the following lower bound holds
\begin{equation}
\perm(p)\geq Z_f^{(\gamma)}(\beta|p)\frac{n!}{n^n}\prod_{(i,j)\in{\cal E}}(1-\beta_{ij})^{\gamma(1-\beta_{ij})}.
\label{new-lower-bound}
\end{equation}
\end{corollary}

Next, one arrives at the following fractional generalization of Proposition \ref{prop:bound5}.
\begin{corollary}[fractional upper bound $\#1$]
\label{corr:fractional-upper}
For any non-negative $p$ and $\beta\in {\cal B}_p^{(int)}$ solving Eqs.~(\ref{m3}) (if the solution exists), the following upper bound holds
$$
\perm(p)\leq Z_f^{(\gamma)}(\beta|p)
(\prod_{(i,j)\in{\cal E}}(1-\beta_{ij})^{\gamma})\prod_j\sum_i\beta_{ij}(1-\beta_{ij})^{-\gamma}.
$$
\end{corollary}
This upper bound follows from combining Theorem \ref{theorem:m-exact}, with the simple (and standard) upper bound,
$\perm(A)\leq \prod_j (\sum_i A_{i,j})$ applied to $A=\beta.*(1-\beta).^{-\gamma}$.

Note that Corollary \ref{corr:fractional-upper}, applied to the $\gamma=0$ case and reinforced by the observation, that for $\gamma\geq 0$ the minimum of the fractional functional (\ref{m1}) is achieved in $\beta\in {\cal B}_p^{(int)}$, translates into
\begin{eqnarray}
\perm(p)\leq Z^{(\gamma=0)}_{o-f}(p).
\label{gamma=0-upper bound}
\end{eqnarray}
Combined with Proposition \ref{prop:gamma_mon}, Eq.~(\ref{gamma=0-upper bound}) results in the following:
\begin{corollary}[fractional upper bound $\#2$]
For any non-negative $p$
\label{corr:gamma=>0-upper bound}
$$\forall \gamma\geq 0:\ \ \perm(p)\leq Z^{(\gamma)}_{o-f}(p). $$
\end{corollary}

This completes the list of inequalities we were able to derive generalizing the BP and MF inequalities stated in the preceding subsection for the fractional case. These generalizations are valid for any $\gamma\in [0;1]$. Therefore, one may hope to derive somewhat stronger statements reinforcing the continuous family of inequalities with the mononotonicity of the fractional approach stated in Proposition \ref{prop:gamma_mon}.

Indeed, combining Eqs.~(\ref{B1}) with Propositions \ref{prop:gamma_mon},\ref{corr:gamma=>0-upper bound} one arrives at
\begin{proposition}[Special $\gamma_*$]
\label{prop:FractionalOptim}
For any non-negative $p$ there exists a special $\gamma_*\in[-1;0]$, such that $\perm(p)=Z_{o-f}^{(\gamma_*)}(p)$, and the minimal fractional solution upper (lower) bounds the permanent at $0\geq \gamma>\gamma_*$ ($-1\leq\gamma<\gamma_*$).
\end{proposition}
Proposition \ref{prop:FractionalOptim} motivates our experimental analysis of the $\gamma_*(p)$ dependence discussed in Section \ref{sec:Experiments}.

Note also that due to the monotonicity stated in Proposition \ref{prop:gamma_mon},  the $\gamma=0$ upper bound on the permanent is tighter than the MF, $\gamma=1$,  upper bound. However, and as discussed in more details in the next subsection, even the $\gamma=0$ upper bound on the permanent is not expected to be tight.

\subsection{Conjectures}
\label{subsec:conject}

It was conjectured in \citet{10Von} that
\begin{eqnarray}
\perm (p)\leq Z_{o-BP}(p)*f(n),
\label{conj1}
\end{eqnarray}
and also that $f\sim \sqrt{n}$.  The second part of the conjecture was disproved by \citet{11Gur} with an explicit counter-example. The inequality in Eq.~(\ref{conj1}) turns into the equality $f(n)=\sqrt{2}^n$ when $p$ is doubly stochastic and block diagonal, with all the elements in the $2\times 2$ blocks equal to $1/2$ \footnote{Note that this special form of the $2\times 2$ block corresponds to the ``double degeneracy" discussed in the paragraph preceding Corollary \ref{cor:BP-uniqness}.}. Then it was conjectured in \citet{11Gur} that
\begin{conjecture}[BP upper bound \cite{11Gur}]
For any non-negative $p$,  $f(n)$ in Eq.~(\ref{conj1}) is $\sim
\sqrt{2}^n$.
\label{conj:Gur}
\end{conjecture}
Another related (but not identical) conjecture of \citet{11Gur} is as follows:
\begin{conjecture}
\label{conj:Gur2}
The following inequality holds for any doubly stochastic $n\times n$ matrix $\phi$:
\begin{eqnarray}
\perm(\phi)\leq \sqrt{2}^n\prod_{(i,j)}(1-\phi_{ij})^{(1-\phi_{ij})}.
\label{conj2}
\end{eqnarray}
\end{conjecture}
Note that if Eq.~(\ref{conj2}) is true it implies according to
\citet{98LSW} a deterministic polynomial-time algorithm to approximate the permanent of $n\times n$ nonnegative matrices within the relative factor $\sqrt{2}^n$.

It can be verified directly that the special matrix (with ``doubly degenerate" blocks) for which the condition (\ref{conj2}) is achieved (i.e., inequality is turned into equality), and it also results
in $Z_f^{(\gamma)}$ with $\gamma=-1/2$ on the right-hand side  of Eq.~(\ref{conj2}).  Therefore
one reformulates Conjecture \ref{conj:Gur} as
\begin{conjecture}
\label{half-conj}
The following inequality holds for any non-negative $p$
\label{conj:upper}
$$
\perm(p)\leq Z_{o-f}^{\gamma=-1/2}(p).
$$
\end{conjecture}
We refer an interested reader to \citet{11Von} for discussion of some other conjectures related to permanents.

\section{Experiments}

\label{sec:Experiments}

We experiment with deterministic and random (drawn from an ensemble) non-negative matrices.

Our simple deterministic example is of the matrices with elements taking two different values such that all the diagonal  and all the off-diagonal elements are the same \cite{10WC}.

In our experiments with stochastic matrices we consider the following four different ensembles
\begin{itemize}
\item {\it $(\lambda_{\mathrm{in}},\lambda_{\mathrm{out}})$:} Ensemble of matrices motivated by \citet{08CKV,10CKKVZ} and corresponding to a mapping between two consecutive images in 2d flows parameterized  by the vector $\lambda=(a,b,c,\kappa)$, where $\kappa$ is the diffusion coefficient and $(a,b,c)$ stand for the three parameters of the velocity gradient tensor (stretching, shear and rotation, respectively -- see \citet{10CKKVZ} for details).  In generating such a matrix $p$ we need to construct two sets of $\lambda$ parameters. The first one, $\lambda_{\mathrm{in}}$, is used to generate an instance of particle positions in the second image, assuming that particles are distributed uniformly at random in the first image. The second one, $\lambda_{\mathrm{out}}$, corresponds to an instance of the guessed values of the parameters in the learning problem, where computation of the permanent is an auxiliary step. (Actual optimal learning consists in computing the maximum of the permanent over $\lambda_{\mathrm{out}}$.) In our simulations we test the quality of the permanent approximations in the special case, when $\lambda_{\mathrm{in}}=\lambda_{\mathrm{out}}$,  and also in other cases when the guessed values of the parameters do not coincide with the input ones, $\lambda_{\mathrm{in}}\neq\lambda_{\mathrm{out}}$.

\item {\it $[0;\rho]$-uniform:} In this case one generates elements of the matrix independently at random and distributed uniformly within the $[0;\rho]$-range.

\item {\it $\delta$-exponential:} In this case one generates elements of the matrix independently at random. Any element is an exponentially distributed random variable with mean $\delta$.

\item {\it $[0;\rho]$-shifted:} We generate the block diagonal matrix with $\left(\begin{array}{cc} 1/2 & 1/2 \\ 1/2 & 1/2\end{array}\right)$ blocks and add independent random and uniformly distributed in the $[0;\rho]$ interval components to all elements of the matrix. The choice of this ensemble is motivated by the special role played by the (doubly degenerate) block-diagonal matrix in the Gurvits conjecture discussed in Section \ref{subsec:conject}.
\end{itemize}

To make the task of the exact computation of the permanent of a random matrix tractable we consider sparsified versions of the ensembles defined above. To achieve this goal we either prune full matrix from the bare (i.e., not yet pruned) ensemble, according to the procedure explained in Appendix \ref{sec:pruning},  or in the case of the $[0;\rho]$-uniform ensemble we first generate a sufficiently sparse sub-graph of the fully connected bipartite graph (for example picking a random subgraph of fixed $O(1)$ degree) and then generate nonzero elements corresponding only to edges of the sub-graph.

\subsection{Deterministic Example}

\begin{figure}
\centering
\subfigure[$\log(\perm(p)/\mbox{Right hand side of Eq.~(\ref{new-lower-bound})})$ vs. $T$ at $n=20$, $w=2$ and different values of $\gamma$]{\includegraphics[width=0.45\textwidth]{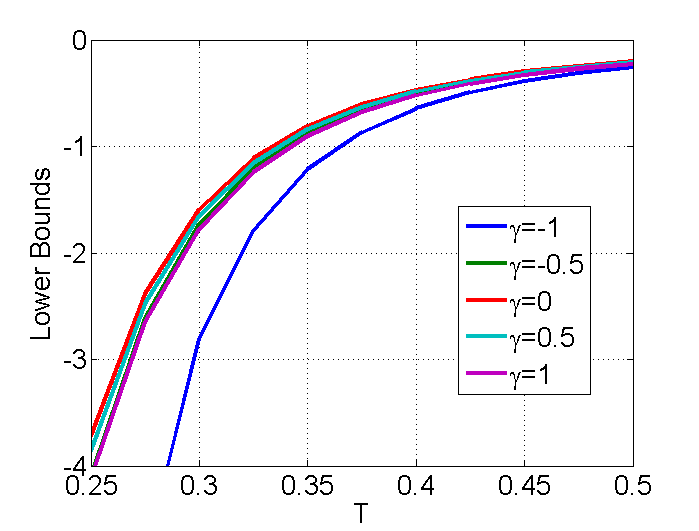}}
\subfigure[$\gamma_*$ vs. $n$ at different values of $T$ and $w=2$]{\includegraphics[width=0.5\textwidth]{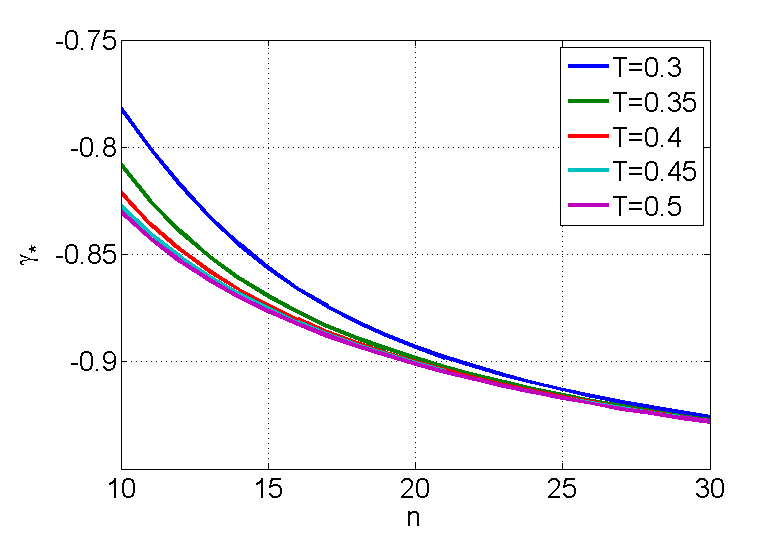}}
\caption{Illustration for the case of the deterministic matrix (\ref{p-det}). Fig (a) shows the gap between the exact permanent and its lower bound estimate by Eq.~(\ref{corr:fractional-low}). Fig. (b) shows dependence of the special $\gamma_*$ on the parameters.}
\label{fig:det_example}
\end{figure}

We consider a simple example which was already discussed in \citet{10WC}. The permanent of the matrix $p$ with elements
\begin{equation}
p_{ij}=\left\{\begin{array}{cc}
w^{1/T},& i=j\\
1,& i\neq j
\end{array}\right.,
\label{p-det}
\end{equation}
where $w>1$ and $T>0$, can be evaluated through the recursion,
$$\sum_{k=0}^n W^{(n-k)/T}\left(\begin{array}{c} n\\ k \end{array}\right)D_k, $$
$D_0=1$, $D_1=0$, and $\forall k\geq 2, D_k=(k-1)(D_{k-1}+D_{k-2})$. On the other hand, seeking for solution of the fractional Eqs.~(\ref{m3}) in the form of a doubly stochastic $\beta$, where
\begin{equation}
\beta_{ij}=\left\{\begin{array}{cc}
1-\epsilon(n-1),& i=j\\
\epsilon,& i\neq j
\end{array}\right.,
\label{eps-eq}
\end{equation}
one finds that $\epsilon$ should satisfy the following transcendental equation,
$$(1-\epsilon(n-1))(1-\epsilon)^\gamma=w^{1/T}(n-1)^\gamma\epsilon^{1+\gamma}.$$ At $T\to\infty$ this equation has a unique uniform, $\epsilon\to 1/(n-1)$, solution. An interior, $\epsilon>0$, solution of Eq.~(\ref{eps-eq}) exists, and it is also unique, at any finite $T$ for $\gamma>-1$. According to \citep{10WC}, the interior solution does not exists at $\gamma=-1$ and $T<\log\omega/\log(n-1)$.

To test the gap between the exact expression for the permanent and the fractional lower bound of Corollary \ref{corr:fractional-low}, we fix $w=2, n=20$ and vary the temperature parameter, $T$. The results are shown in Fig.~\ref{fig:det_example}a. One finds that the gap depends on $\gamma$ with $\gamma=0$ resulting in the best lower bound for all the tested temperatures.  One also observes that the $\gamma$-dependence of the gap decreases with increase in $T$. Fig.~\ref{fig:det_example}b shows dependence of the special $\gamma_*$, defined in Proposition \ref{prop:FractionalOptim}, on $n$ and $T$ at $w=2$. One finds that, consistently with the Conjecture \ref{half-conj}, $\gamma_*$ is always smaller than $-1/2$ and it also decreases with increase in either $n$ or $T$.

\subsection{Random Matrices. Special $\gamma_*$.}

We search for the special $\gamma=\gamma_*$, defined in Proposition \ref{prop:FractionalOptim}, by calculating the permanent of a full matrix, $p$,
of size $n \times n$, with $n = 3, \ldots, 14$, and of a pruned matrix with $n=10,\ldots,40$, and then comparing it with the fractional value $Z_{f}^{(\gamma)}(\beta,p)$ \footnote{In the following we will use the shorter notation, $Z_f^{(\gamma)}$ for this object.}, where the doubly stochastic $\beta$ solves Eqs.~(\ref{m3}) for given $p$, for different $\gamma$. By repeatedly evaluating the fractional approximation for different values of $\gamma$ and then taking advantage of the $Z_f^{(\gamma)}$ monotonicity and continuity with respect to $\gamma$ and performing a search we find the special $\gamma$ for a specific $p$.

In general we observed that the special $\gamma_*$ for tested matrices was always less than or equal to $-1/2$, which is consistent with Conjecture \ref{conj:upper}. We also observed, estimating or extrapolating the approximate value of the special $\gamma_*$ for a given matrix, that it might be possible to estimate the permanent of a matrix efficiently and very accurately for some ensembles.

\subsubsection{The $(\lambda_{\mathrm{in}},\lambda_{\mathrm{out}})$ ensemble}

\begin{figure}
\centering
\subfigure[$\lambda_{\mathrm{in}}=$$ \lambda_{\mathrm{out}}=(1,1,1,1)/2$]{\includegraphics[width=0.3\textwidth]{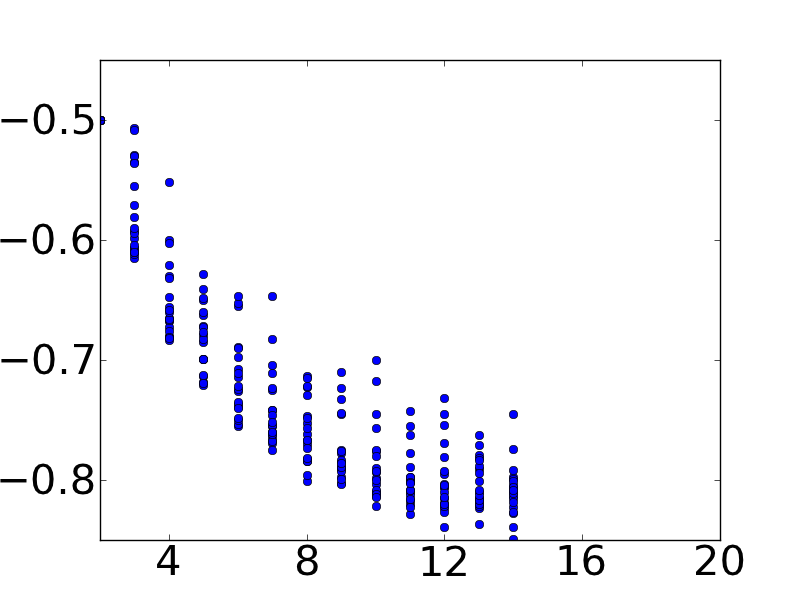}}
\subfigure[$\lambda_{\mathrm{in}}=$$
\lambda_{\mathrm{out}}=(2,2,2,1/2)$]{\includegraphics[width=0.3\textwidth]{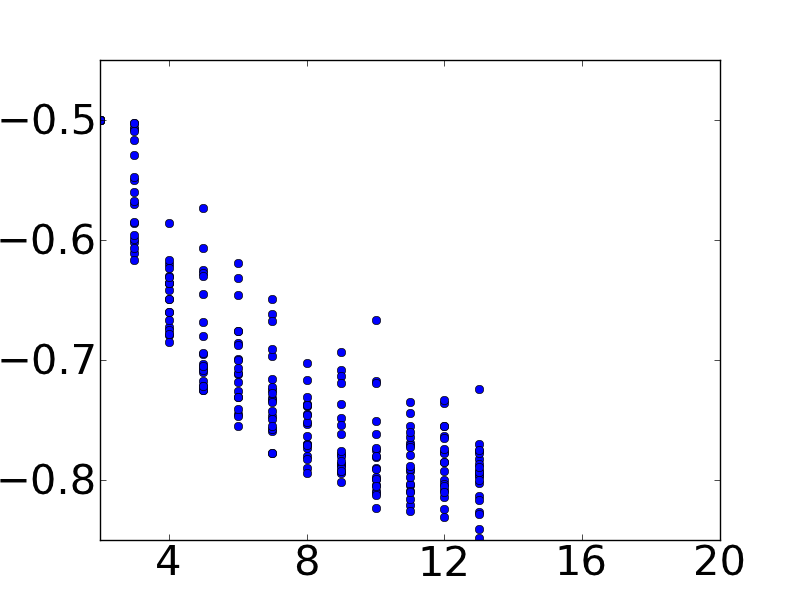}}\\
\subfigure[$\lambda_{\mathrm{in}}=$$
\lambda_{\mathrm{out}}=(0,0,0,1)$]{\includegraphics[width=0.3\textwidth]{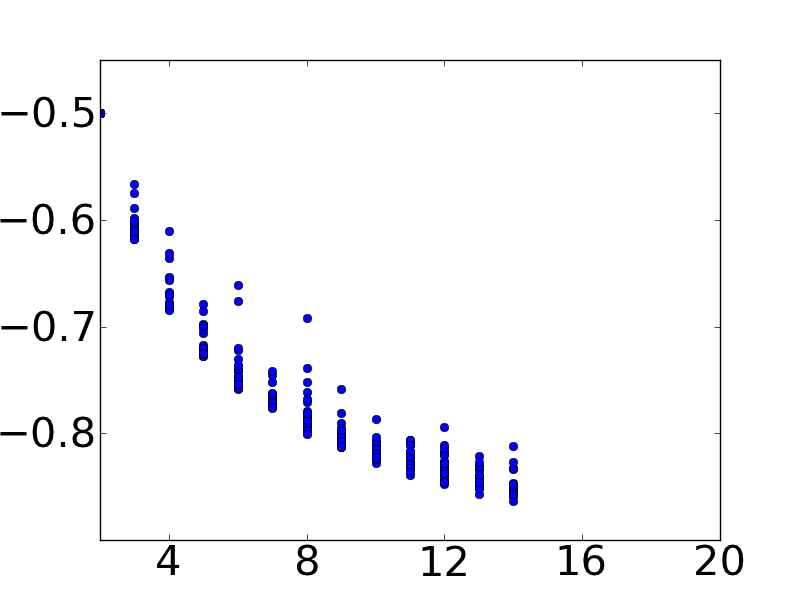}}
\subfigure[$\lambda_{\mathrm{in}}=$$
\lambda_{\mathrm{out}}=(1,1,1,1)$]{\includegraphics[width=0.3\textwidth]{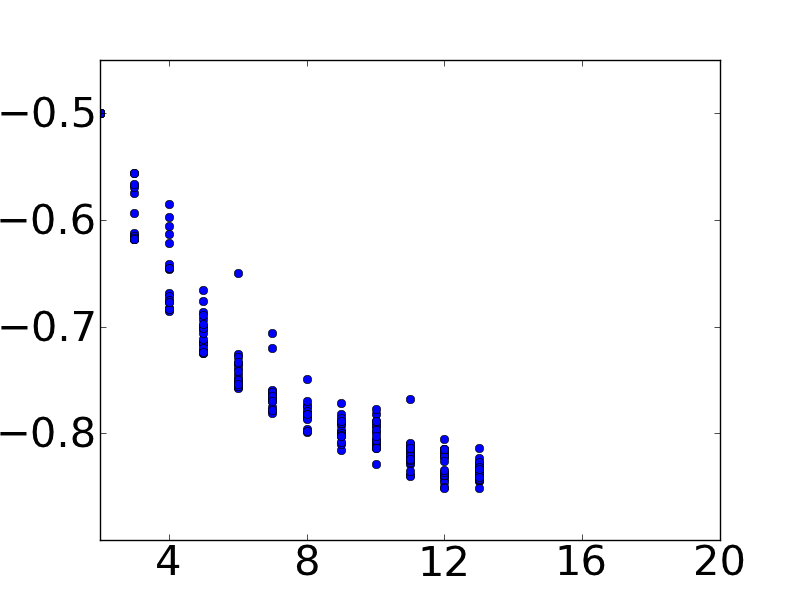}}
\subfigure[$\lambda_{\mathrm{in}}=$$
\lambda_{\mathrm{out}}=(2,2,2,1)$]{\includegraphics[width=0.3\textwidth]{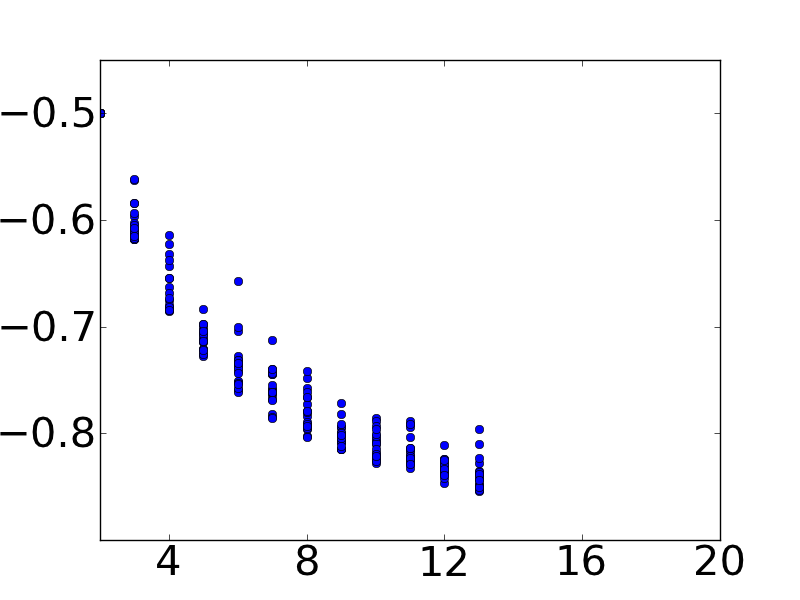}}\\
\subfigure[$\lambda_{\mathrm{in}}=$$
\lambda_{\mathrm{out}}=(1,1,1,1)/4$]{\includegraphics[width=0.3\textwidth]{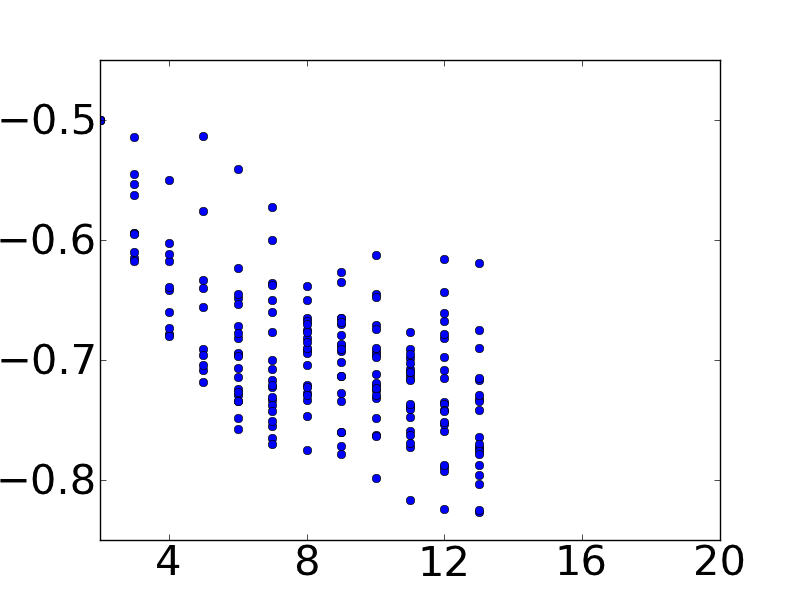}}
\subfigure[$\lambda_{\mathrm{in}}=$$
\lambda_{\mathrm{out}}=(0,0,0,1/10)$]{\includegraphics[width=0.3\textwidth]{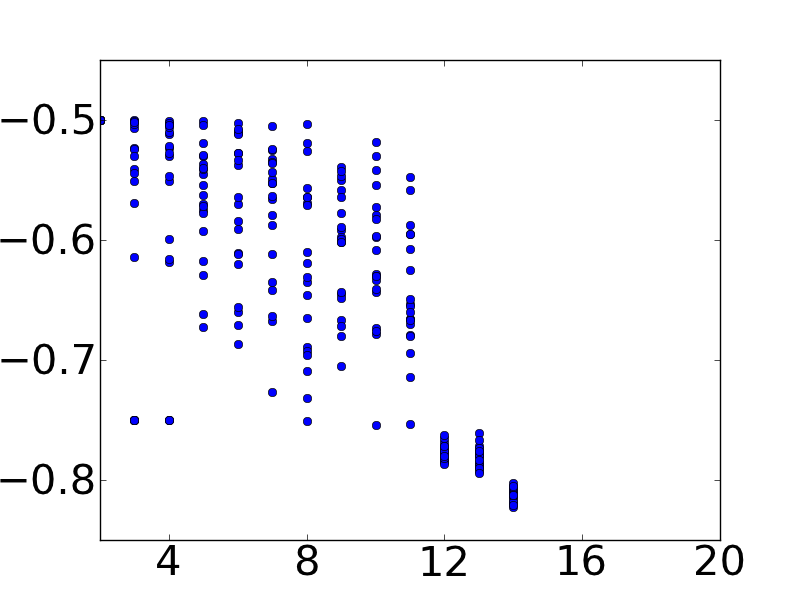}}
\subfigure[$\lambda_{\mathrm{in}}=$$
\lambda_{\mathrm{out}}=(0,0,0,2)$]{\includegraphics[width=0.3\textwidth]{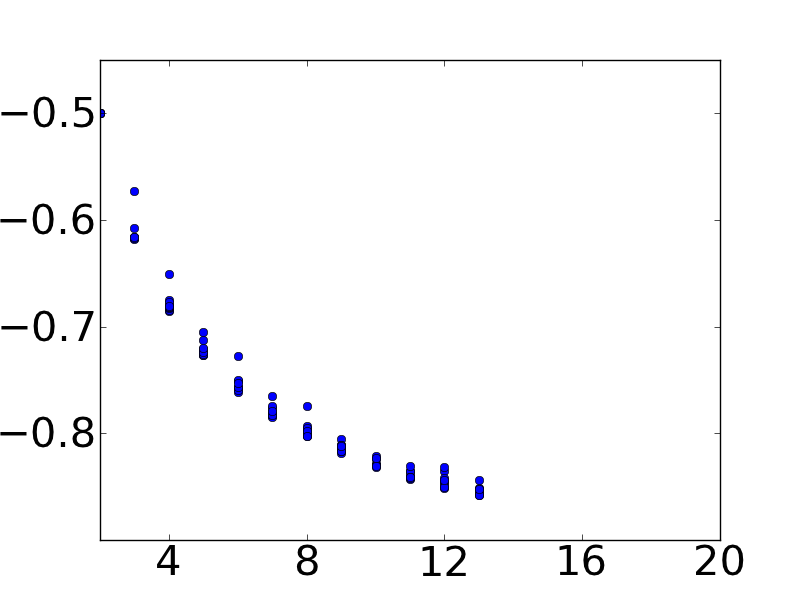}}\\
\subfigure[$\lambda_{\mathrm{in}} = (1,1,1,1)$,\hspace{1cm}
$\lambda_{\mathrm{out}}=(0,0,0,1)$]{\includegraphics[width=0.3\textwidth]{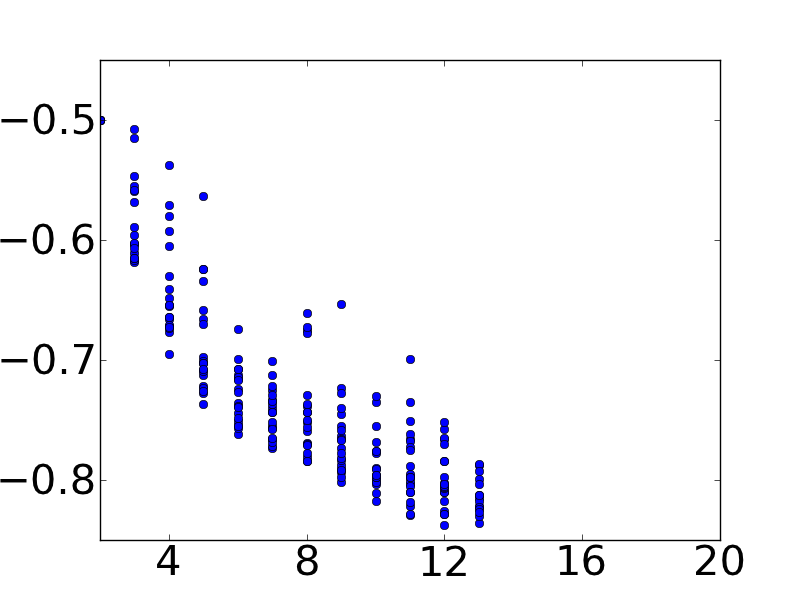}}
\subfigure[$\lambda_{\mathrm{in}} = (0,0,0,2)$,\hspace{1cm} $\lambda_{\mathrm{out}}=(0,0,0,1)$]{\includegraphics[width=0.3\textwidth]{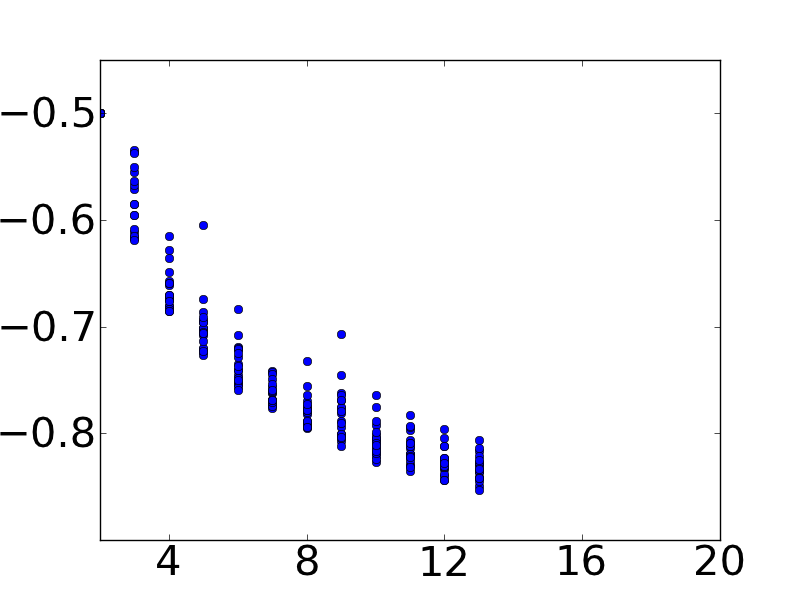}}
\subfigure[$\lambda_{\mathrm{in}} = (0,0,0,1)$,\hspace{1cm} $\lambda_{\mathrm{out}}=(0,0,0,2)$]{\includegraphics[width=0.3\textwidth]{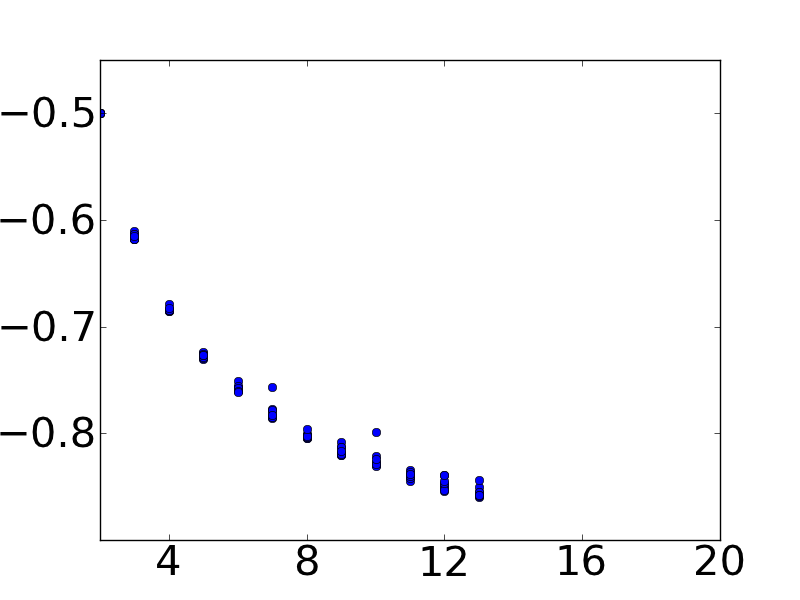}}
\caption{
Scatter plot of the special $\gamma_*$ calculated for instances from the $(\lambda_{\mathrm{in}},\lambda_{\mathrm{out}})$ ensemble and varying the matrix size within the $2\div 14$ range (no pruning).
}
\label{fig:gamma_small}
\end{figure}

\begin{figure}
\centering
\subfigure[$\lambda_{\mathrm{in}} = \lambda_{\mathrm{out}}=(1,1,1,1)/2$]{\includegraphics[width=0.3\textwidth]{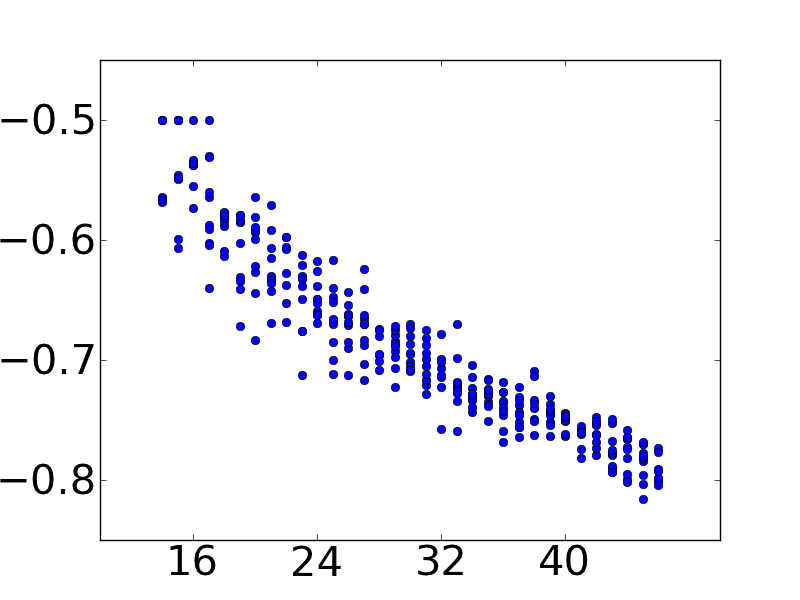}}
\subfigure[$\lambda_{\mathrm{in}} = \lambda_{\mathrm{out}}=(1,1,1,1)$]{\includegraphics[width=0.3\textwidth]{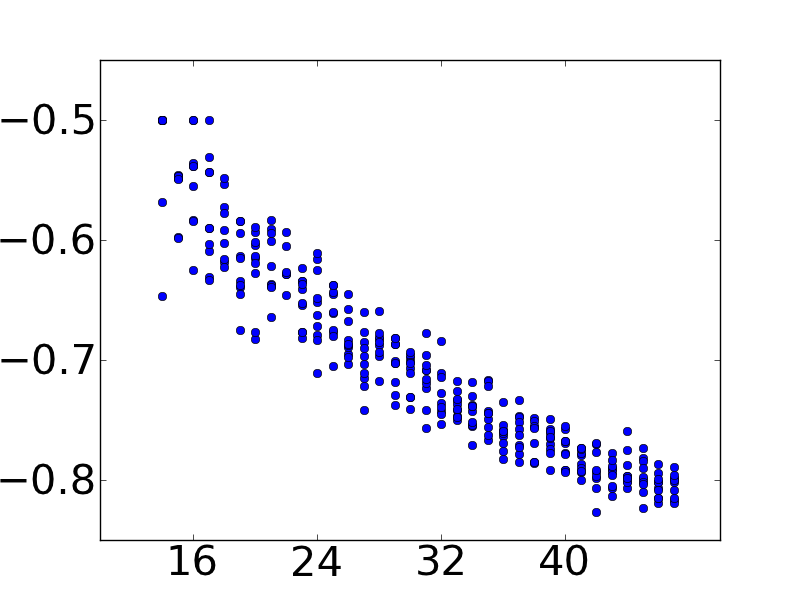}}
\subfigure[$\lambda_{\mathrm{in}} = \lambda_{\mathrm{out}}=(0,0,0,1)$]{\includegraphics[width=0.3\textwidth]{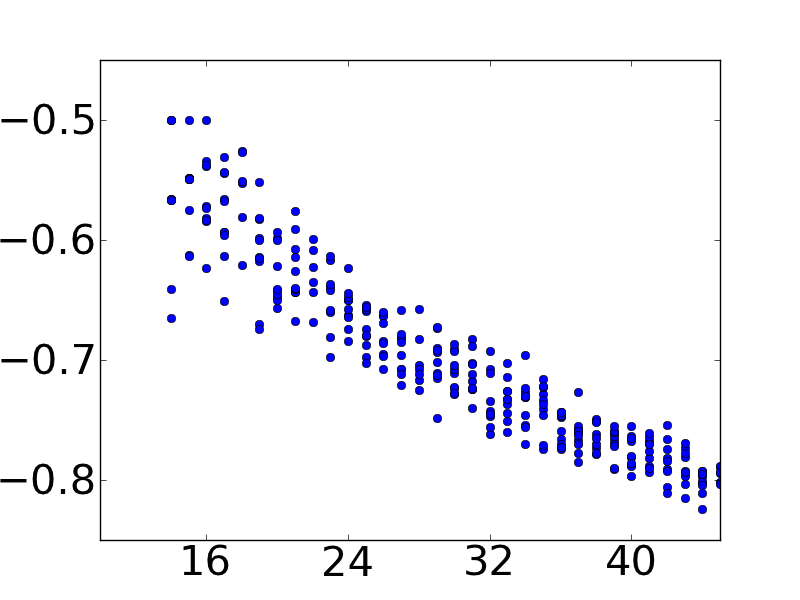}}
\subfigure[$\lambda_{\mathrm{in}} = \lambda_{\mathrm{out}}=(2,2,2,1/2)$]{\includegraphics[width=0.3\textwidth]{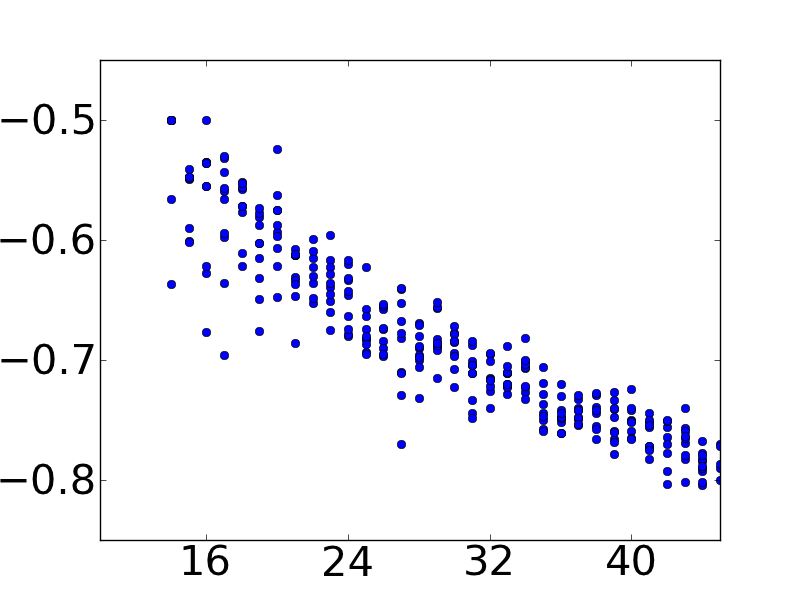}}
\subfigure[$\lambda_{\mathrm{in}} = (0,0,0,2)$, $\lambda_{\mathrm{out}}=(0,0,0,1)$]{\includegraphics[width=0.3\textwidth]{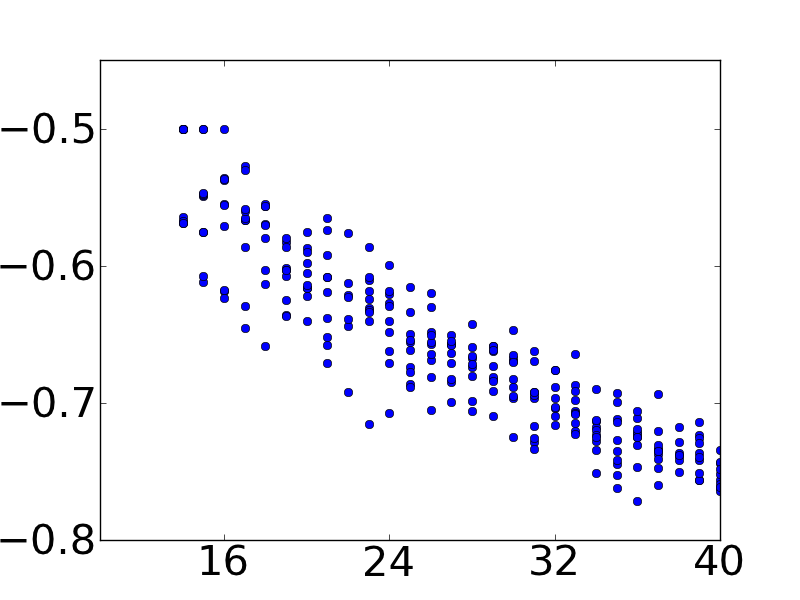}}
\subfigure[$\lambda_{\mathrm{in}} = (1,1,1,1), \lambda_{\mathrm{out}}=(0,0,0,1)$]{\includegraphics[width=0.3\textwidth]{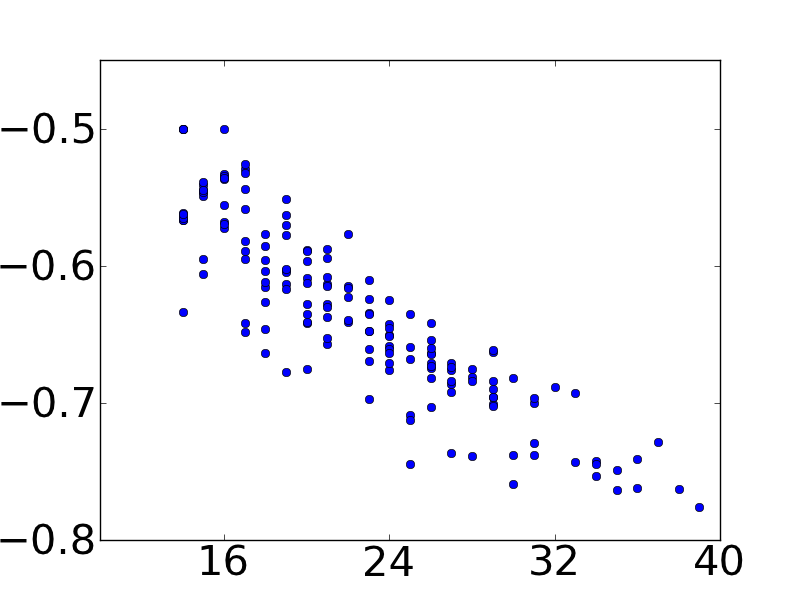}}
\caption{
Scatter plot of the estimated special $\gamma_*$ calculated for instances from the $(\lambda_{\mathrm{in}},\lambda_{\mathrm{out}})$ ensemble and varying the matrix size within the $15\div 40$ range (with 90\% pruning).
}
\label{fig:gamma_large}
\end{figure}

In this subsection we describe  experiments with several of the $(\lambda_{\mathrm{in}},\lambda_{\mathrm{out}})$ ensembles. We are interested in studying the dependence of the special $\gamma_*$, defined in Proposition \ref{prop:FractionalOptim}, on the matrix size and other parameters of the ensemble. We consider here a variety of cases.

Fig.~\ref{fig:gamma_small} shows the results of experiments with full but (relatively) small matrices and different values $\lambda_{\mathrm{in}},\lambda_{\mathrm{out}}$. The results are presented in the form of a scatter plot, showing results for different matrix instances from the same ensemble.

As can be seen from the grouping of the first five plots in Fig.~\ref{fig:gamma_small}, the dependence of the special $\gamma_*$ on the matrix size at $ \lambda_{\mathrm{in}}=\lambda_{\mathrm{out}}$ is largely sensitive to the diffusion parameter $\kappa$ and it is not so dependent on the advection parameters $a,b,c$. Indeed,  Figs.~\ref{fig:gamma_small} (a,b) are similar to each other, as are Figs. ~\ref{fig:gamma_small}(c-e), despite having different values of $a,b,c$.

Figs.~\ref{fig:gamma_small}(a-e), along with Figs.~\ref{fig:gamma_small}(f,g), also demonstrate an interesting feature: the lower $\kappa$, the more erratic the behavior of the special $\gamma_*$, with Figs.~\ref{fig:gamma_small}(f,g) demonstrating this tendency at its extreme.
With low diffusion, matrices were dominated by the largest permutation and search for the special $\gamma_*$ became less meaningful, with seemingly random behavior.

Analyzing the three last cases in  Fig.~\ref{fig:gamma_small} with $\lambda_{\mathrm{in}}\neq\lambda_{\mathrm{out}}$, we observed that the larger the value of $\kappa_{out}$, the more regular the resulting behavior.

Fig.~\ref{fig:gamma_large} shows the same scatter plots as in Fig.~\ref{fig:gamma_small}, observed for larger but sparser (90\% pruned) matrices. We observed the general tendency for the average special $\gamma_*$ to decrease with increasing $n$;  however, it is not clear from the observations if the resulting level of fluctuations decreases with the increase in $n$ or remains the same.

Summarizing, for the $(\lambda_{\mathrm{in}}, \lambda_{\mathrm{out}})$ ensemble, we found that the behavior of the special $\gamma_*$ with respect to matrix size to be largely dependent on $\kappa_{out}$, the diffusion coefficient used to generate the matrix, while the dependance of other factors is significantly less pronounced. The average special $\gamma$ decreases with increasing $n$, while respective variance remains roughly the same.

\subsubsection{Uniform and $\delta$-exponential ensembles}

\begin{figure}
\centering
\subfigure[uniform ensemble. Full (small) matrix.]{\includegraphics[width=0.3\textwidth]{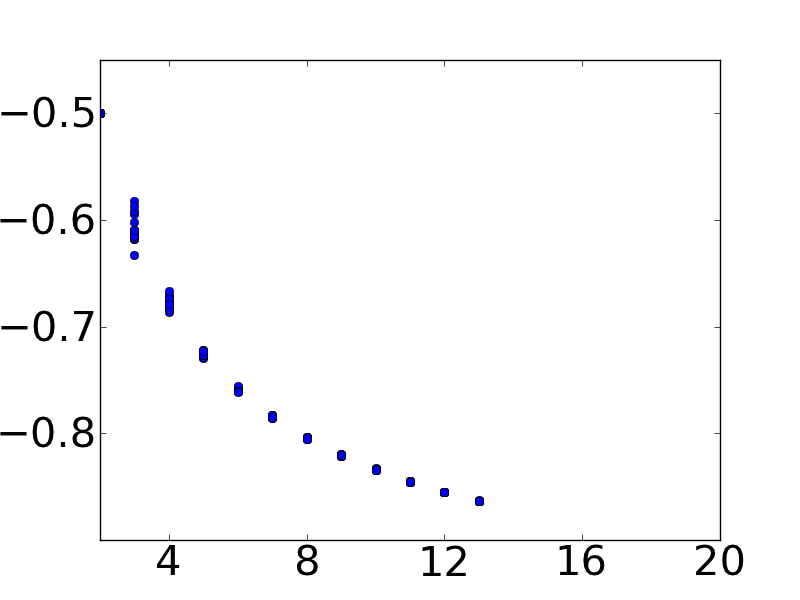}}
\subfigure[$\delta$-exponential ensemble with $\delta=1$. Full (small) matrix.]{\includegraphics[width=0.3\textwidth]{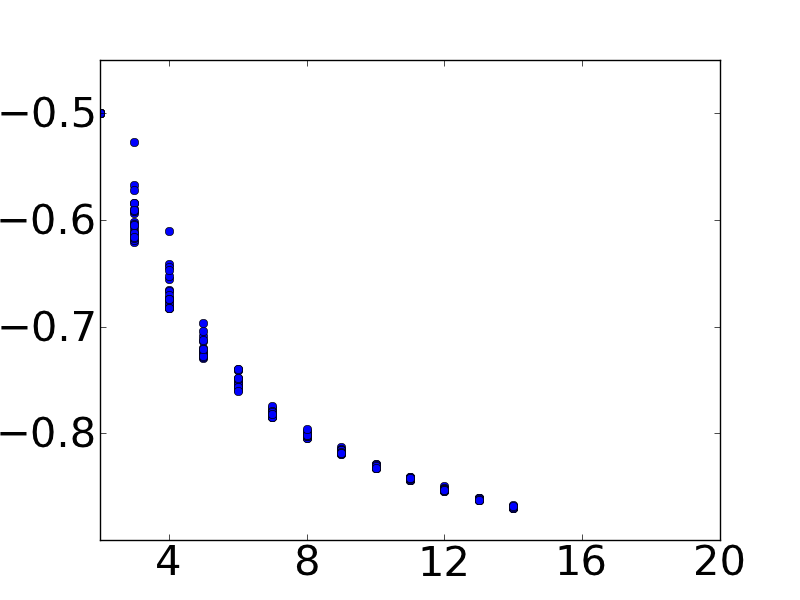}}
\caption{Scatter plot of the estimated special $\gamma$ calculated for instances from a random matrix ensemble and varying the matrix size.}
\label{fig:gamma_random}
\end{figure}

Fig.~\ref{fig:gamma_random} shows scatter plots for examples of the (a) $[0;1]$-uniform ensemble, and (b) $\delta$-exponential ensemble. Here we found a very impressive decrease in variance with increase in the matrix size.
Besides,  we observe that in spite of their difference, the two ensembles show qualitatively similar behavior of $\gamma_*$ as a function of $n$. This indicates that for large matrices, whose entries are independent random variables, we could achieve excellent accuracy by extrapolating on the special $\gamma_*$ and estimating $Z_{f}(\gamma_*)=\perm(p)$. Indeed, the rapidly decreasing variance suggests that most of the error would come from the extrapolation of the average special $\gamma_*$ for a given matrix size and not the error of $Z_f(\gamma_{*,\mathrm{average}}) - Z_f(\gamma_{*,\mathrm{actual}})$, since $\gamma_{*,\mathrm{average}}$ and $\gamma_{*,\mathrm{actual}}$ will be very close in value.

\subsubsection{The $[0;\rho]$-shifted ensemble}

\begin{figure}%
\centering
\subfigure[$\rho=1/20$-shifted]{\includegraphics[width=0.3\textwidth]{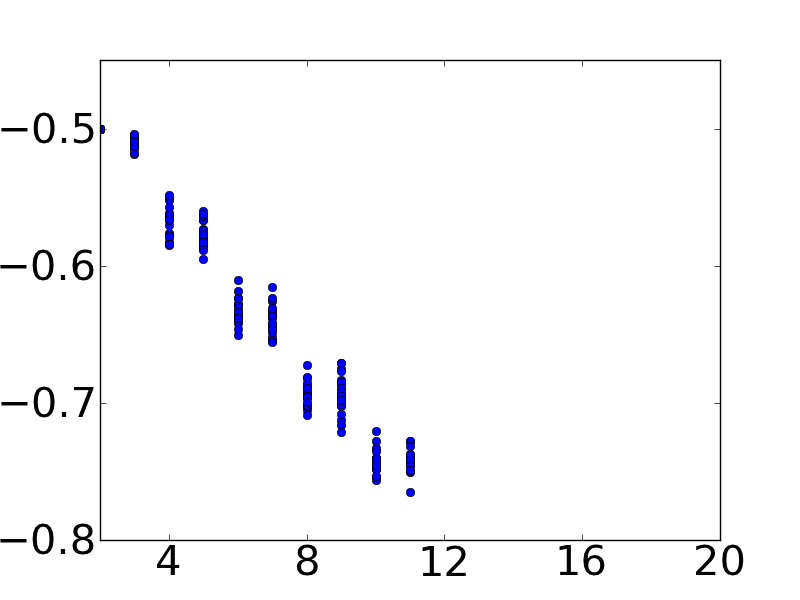}}
\subfigure[$\rho=1/200$-shifted]{\includegraphics[width=0.3\textwidth]{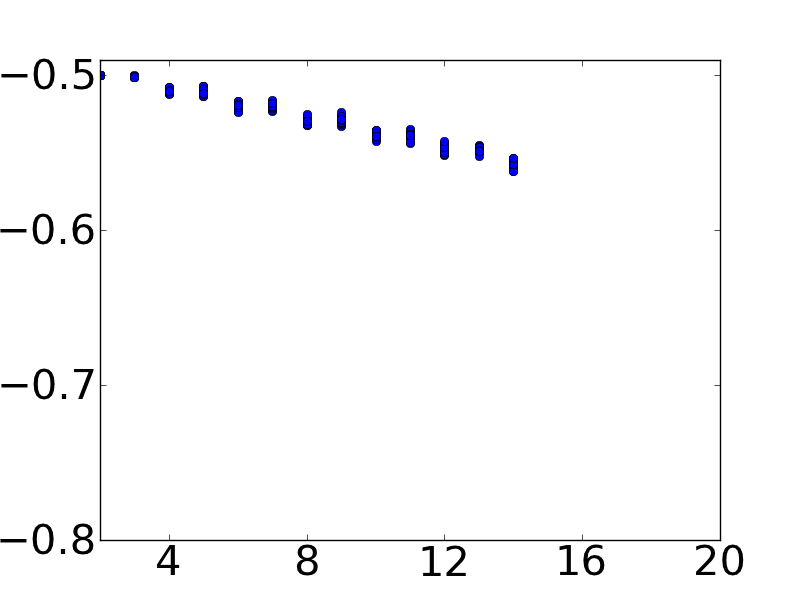}}
\caption{Scatter plot of the estimated special $\gamma_*$ calculated for instances from two examples of the $[0;\rho]$-shifted ensembles and varying the matrix size.}
\label{fig:gamma_shifted}
\end{figure}

We also studied the special $\gamma_*$ in ``badly-behaved'' cases such as the one brought up earlier, with $2\times 2$ squares of $1/2$'s positioned along the diagonal. (See discussion in Section \ref{subsec:conject}.) It can be easily shown that the special value of $\gamma_*$ of the bare block-diagonal matrix is $-1/2$. Unsurprisingly, our experiments, documented in Fig.~\ref{fig:gamma_shifted}, showed that: (a) the resulting $\gamma$ is always smaller than $-1/2$, and (b) as more noise was introduced, the special $\gamma_*$ decreased in value faster with respect to matrix size. However,  this decrease with $n$ towards smaller $\gamma_*$ was much slower than in other ensembles, particularly for low noise.

\subsection{Random Matrices: Testing Inequalities and Conjectures}

\begin{figure}
\centering
\includegraphics[width=\textwidth]{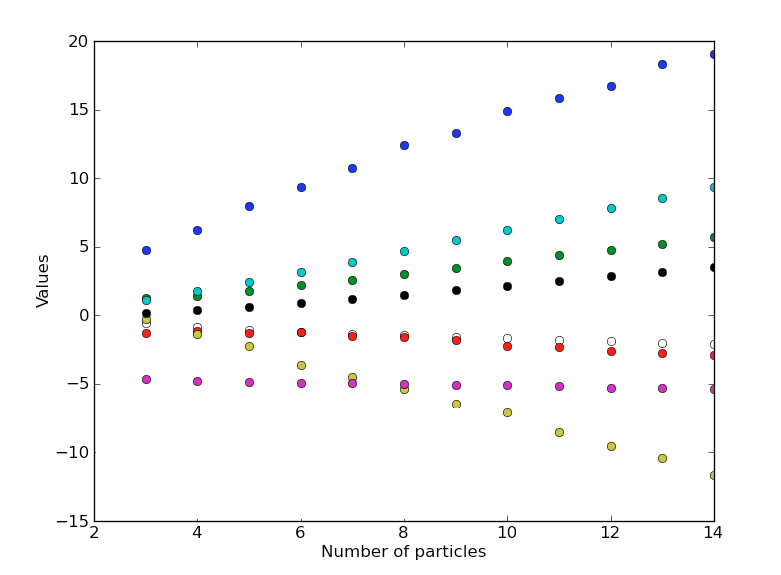} 
\caption{This figure describes the gap between the actual value of the permanent and the values of the various theoretical upper and lower bounds described in the paper and averaged over different simulation trials. Each color corresponds to an upper or lower bound as follows: Blue corresponds to the MF approximation to the permanent, $Z_{MF}$.
Green corresponds to $\log(Z_{BP}(\prod_{(i,j)\in{\cal E}}(1-\beta_{ij}))^{-1}\prod_j(1-\sum_i(\beta_{ij})^2)/Z)$.
Red corresponds to $\log(Z_{BP} \prod_{i,j}(1-\beta_{ij})^{\beta_{ij}-1}\frac{n!}{n^n}/Z)$.
White corresponds to $\log(Z_{BP}/Z)$.
Yellow corresponds to $\log(2 Z_{BP}(\prod_{i,j}(1-\beta_{ij}))^{-1}\prod_i \beta_{ii}(1-\beta_{ii})/Z)$.
Purple corresponds to $\log(0.01 Z_{BP} \sqrt n/Z)$.
Cyan corresponds to $\log(Z_f^{\gamma=0}/Z)$.
Black corresponds to $\log(Z_f^{\gamma=-0.5}/Z)$.
To make a data point, $100$ instances, each corresponding to a new matrix are drawn, and the log of the ratio of the bound to the actual permanent is recorded. The data shown corresponds to matrices from the $(\lambda_{\mathrm{in}}, \lambda_{\mathrm{out}})$ ensemble.}
\label{fig:inequality}
\end{figure}

\begin{figure}
\centering
\subfigure[$\log(Z_{MF}/Z)$]{\includegraphics[width=0.3\textwidth]{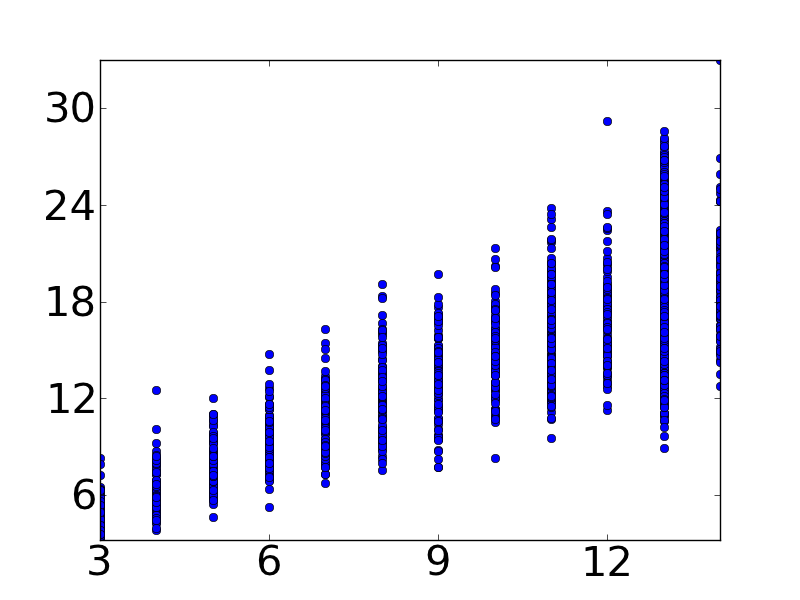}}
\subfigure[$\log((Z_{BP}(\prod_{(i,j)\in{\cal E}}(1-\beta_{ij}))^{-1}\prod_j(1-\sum_i(\beta_{ij})^2))/Z)$]{\includegraphics[width=0.3\textwidth]{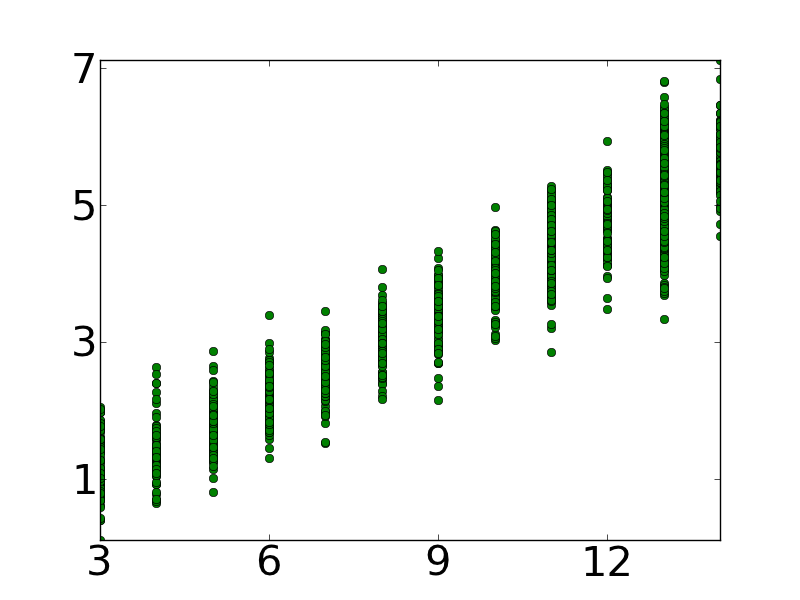}}
\subfigure[$\log((Z_{BP} \prod_{i,j}(1-\beta_{ij})^{\beta_{ij}-1}\frac{n!}{n^n})/Z)$]{\includegraphics[width=0.3\textwidth]{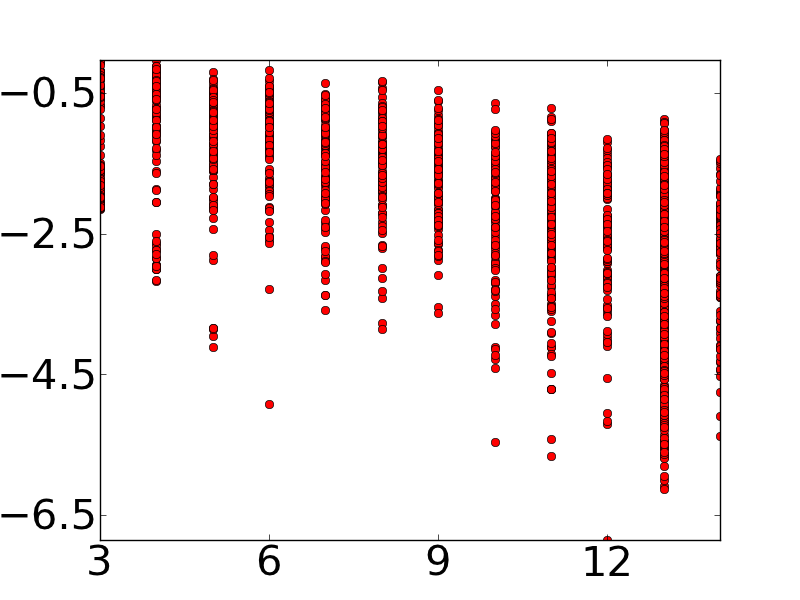}}
\subfigure[$\log(Z_{BP}/Z)$]{\includegraphics[width=0.3\textwidth]{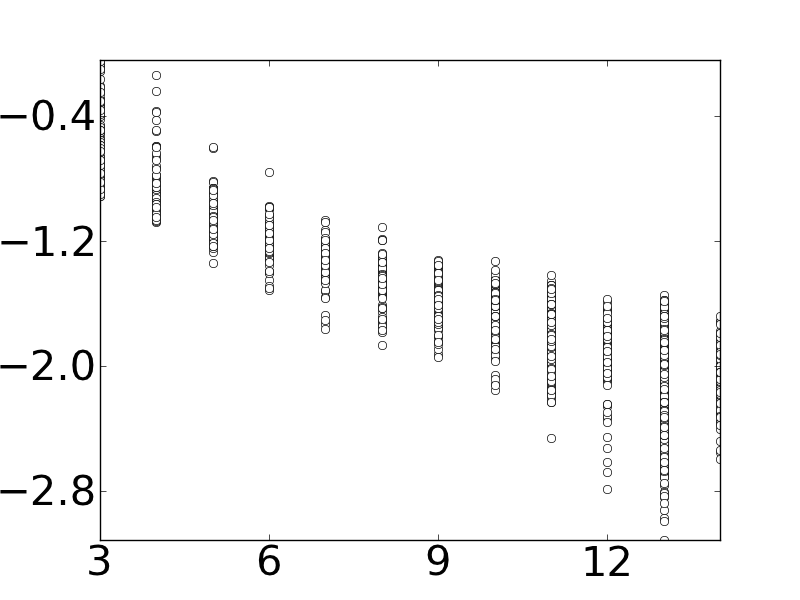}}
\subfigure[$\log((2 Z_{BP}(\prod_{i,j}(1-\beta_{ij}))^{-1}\prod_i \beta_{i\Pi(i)}(1-\beta_{i\Pi(i)}))/Z)$]{\includegraphics[width=0.3\textwidth]{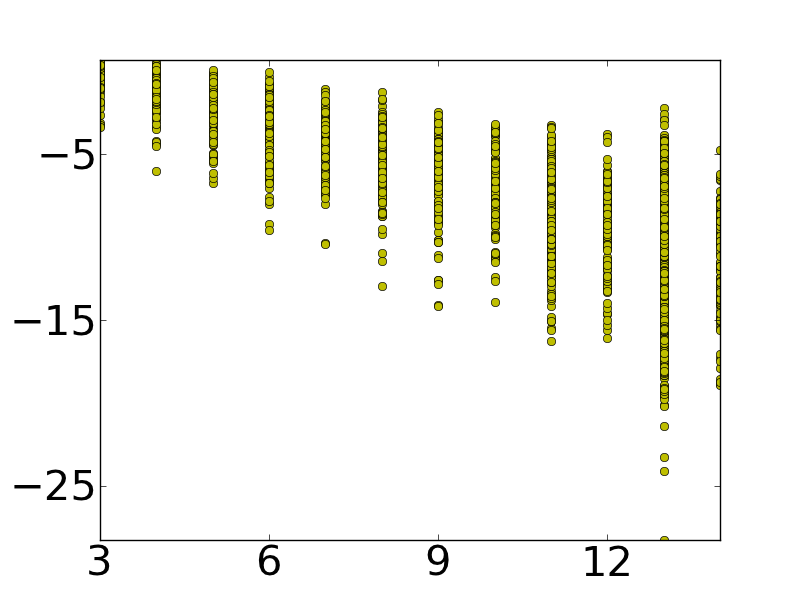}}
\subfigure[$\log(0.01 Z_{BP} \sqrt n/Z)$]{\includegraphics[width=0.3\textwidth]{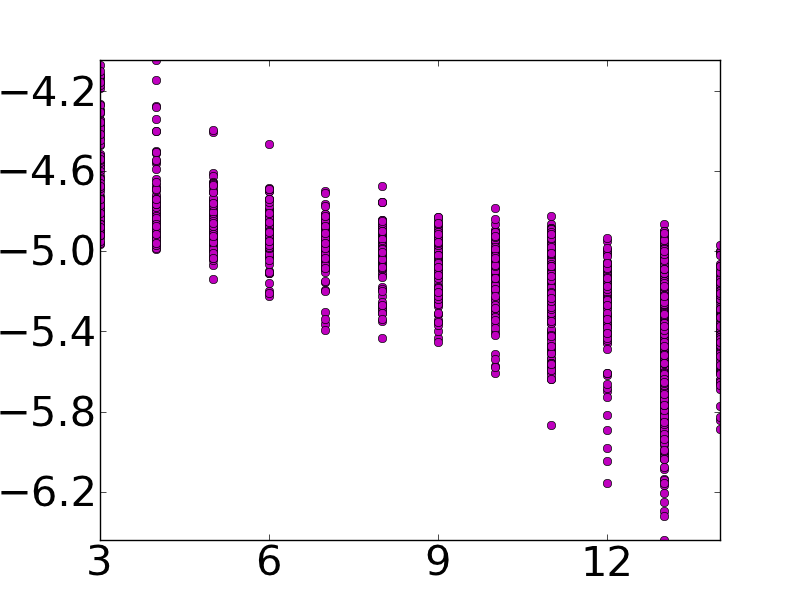}}
\subfigure[$\log(Z_f^{(\gamma=0)}/Z)$]{\includegraphics[width=0.3\textwidth]{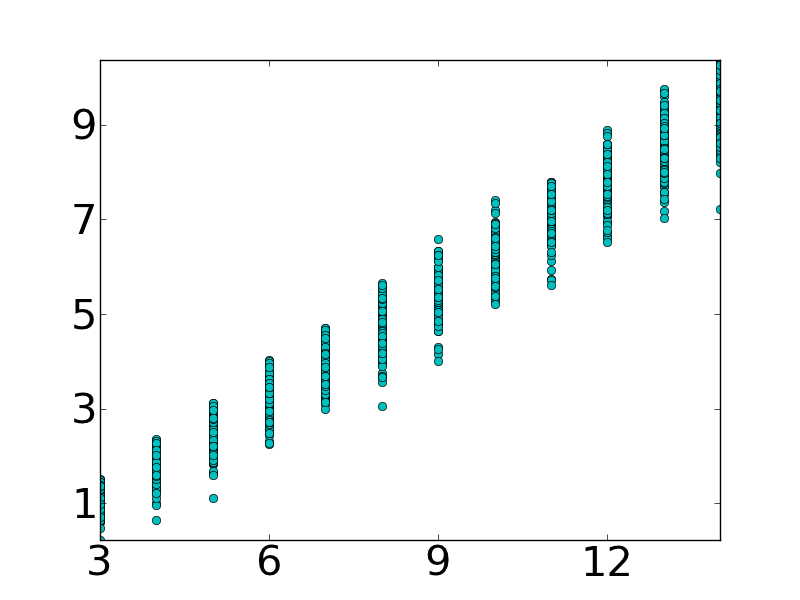}}
\subfigure[$\log(Z_f^{(\gamma=-0.5)}/Z)$]{\includegraphics[width=0.3\textwidth]{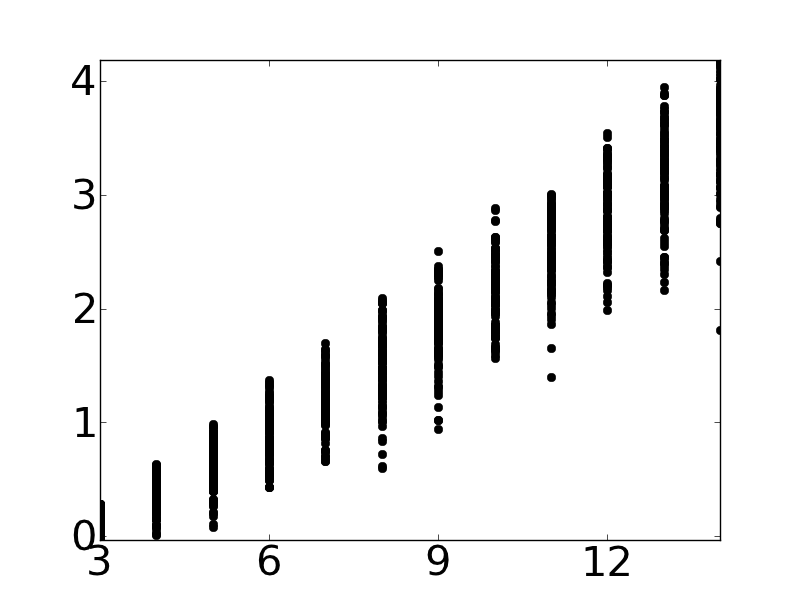}}
\caption{ Scatter plots for the data shown in Figs.~(\ref{fig:inequality}). For better presentation the data is split into 8 sub-figures. The vertical axis of each scatter plot is specific to the behavior of the expression with respect to the matrix size
and the color coding for different objects tested coincides with that used in Fig.~\ref{fig:inequality}.}
\label{fig:inequalityscatter}
\end{figure}

\begin{figure}
\centering
\includegraphics[width=\textwidth]{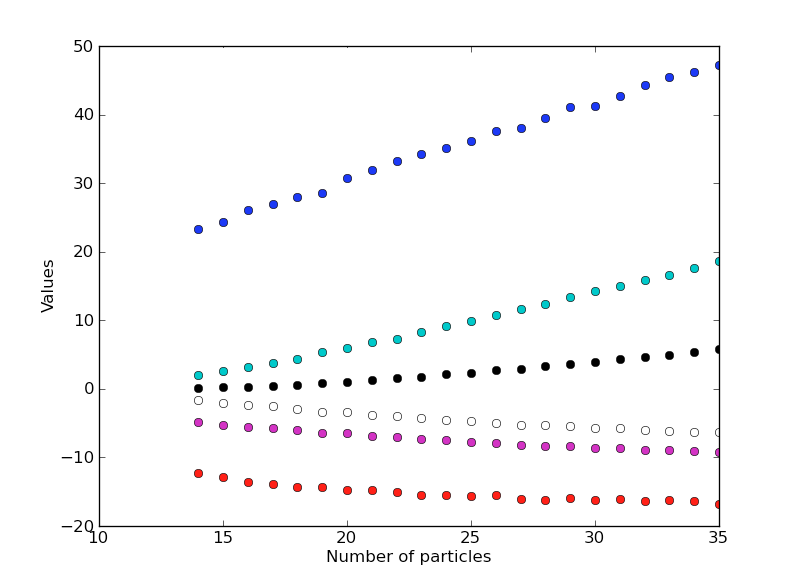} 
\caption{The data is shown like in Fig.~\ref{fig:inequality}, but for large, sparsified matrices. Less meaningful expressions were removed from the plot. Each color corresponds to a mathematical expression, as follows,
blue: $\log(Z_{MF} / Z)$,
red: $\log(Z_{BP} \prod_{i,j}(1-\beta_{ij})^(\beta_{ii}-1) \cdot n!/(n^n Z))$,
white: $\log(Z_{BP}/Z)$,
purple: $\log(Z_{BP}\sqrt{n}/(100Z))$,
cyan: $\log(Z_f(\gamma=0)/Z)$,
black: $\log(Z_f(\gamma=-1/2)/Z)$; where $Z=\perm(p)$.}
\label{fig:prunedinequality}
\end{figure}

\begin{figure}
\centering
\subfigure[$\log(Z_{MF}/Z)$]{\includegraphics[width=0.3\textwidth]{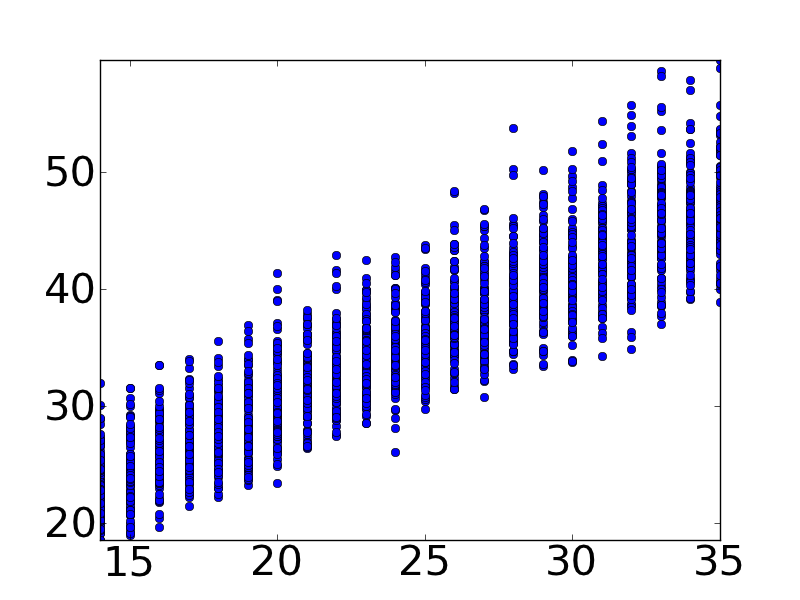}}
\subfigure[$\log((Z_{BP} \prod_{i,j}(1-\beta_{ij})^{\beta_{ij}-1}\frac{n!}{n^n})/Z)$]{\includegraphics[width=0.3\textwidth]{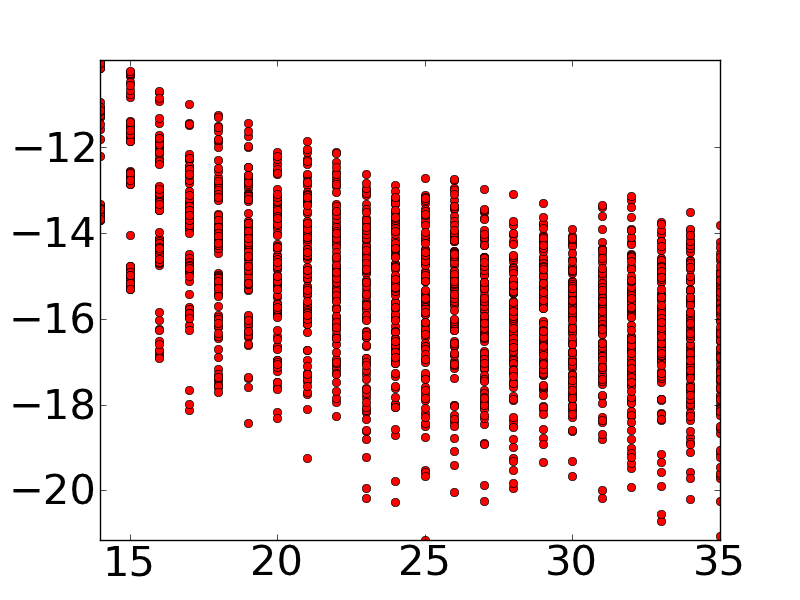}}
\subfigure[$\log(Z_{BP}/Z)$]{\includegraphics[width=0.3\textwidth]{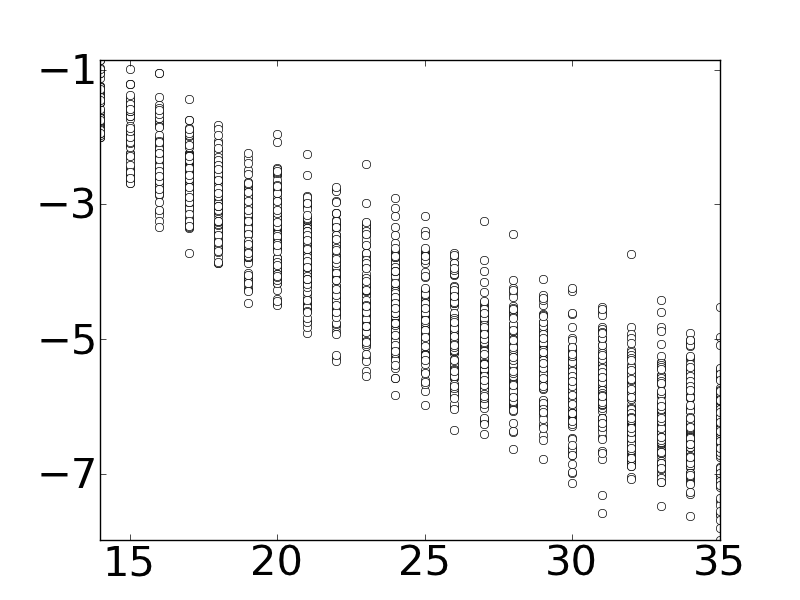}}
\subfigure[$\log(0.01 Z_{BP} \sqrt n/Z)$]{\includegraphics[width=0.3\textwidth]{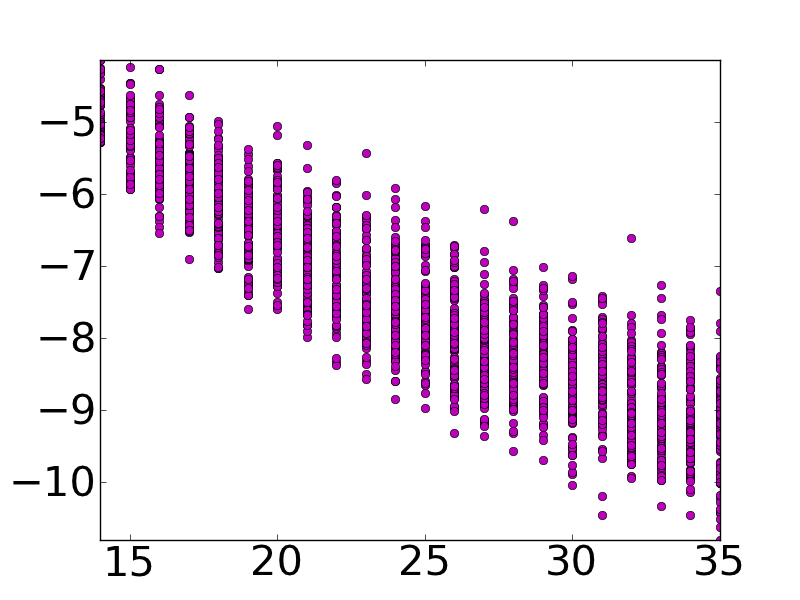}}
\subfigure[$\log(Z_f^{(\gamma=0)}/Z)$]{\includegraphics[width=0.3\textwidth]{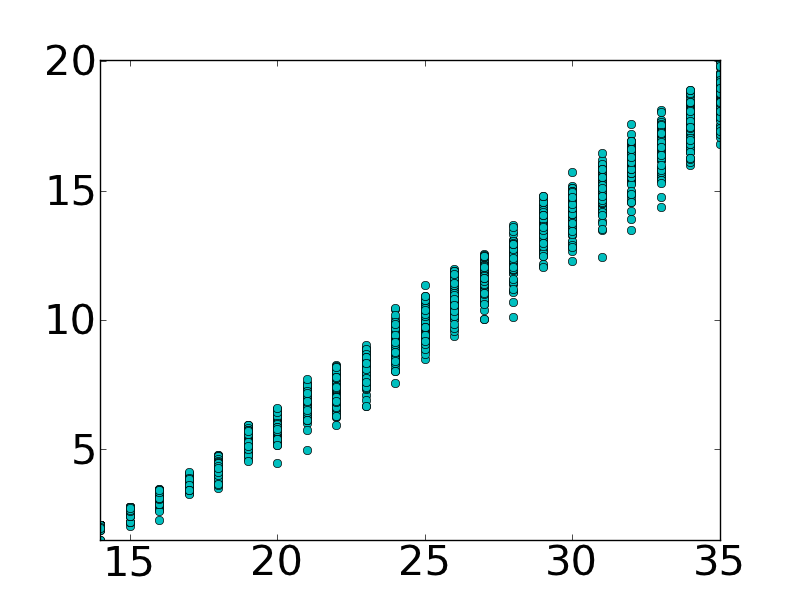}}
\subfigure[$\log(Z_f^{(\gamma=-0.5)}/Z)$]{\includegraphics[width=0.3\textwidth]{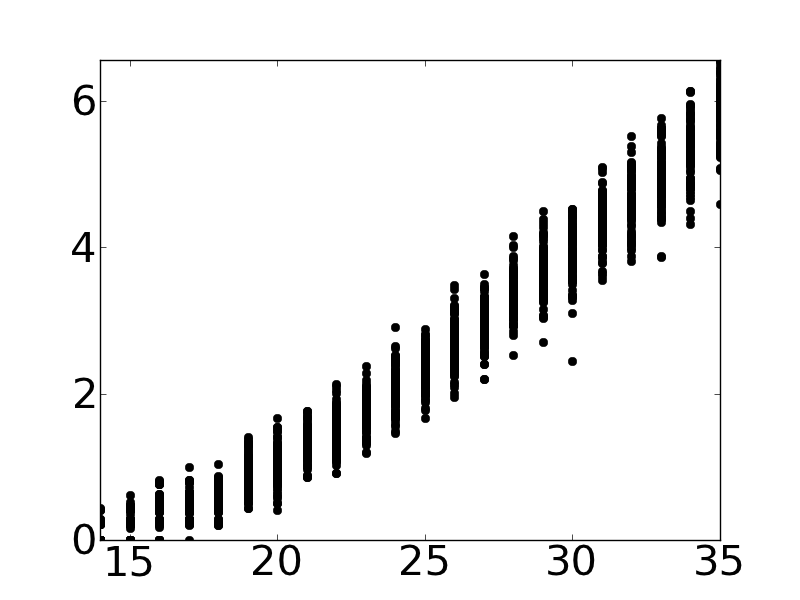}}
\caption{ Scatter plots for the data shown in Figs.~(\ref{fig:prunedinequality}). For better presentation the data is split into 6 sub-figures. Each scatter plot is specific to the behavior of the expression with respect to the matrix size and the color coding for different objects tested coincides with the one used in Fig.~\ref{fig:prunedinequality}.}
\label{fig:prunedinequalityscatter}
\end{figure}

Figs.~\ref{fig:inequality},\ref{fig:inequalityscatter} and Figs.~\ref{fig:prunedinequality},\ref{fig:prunedinequalityscatter},  showing average behavior and scatter plots for smaller and larger (pruned) matrices, respectively (see figure captions for explanations),  present experimental verification to the variety of inequalities discussed in Section \ref{subsec:ineq}. The ensemble used for these plots was the ensemble $\lambda_{\mathrm{in}}=\lambda_{\mathrm{out}}=(1,1,1,1)/2$. The data suggests that neither of the bounds are actually tight, and moreover the values of the gaps, between the exact expression and respective estimates tested, fluctuate more strongly with increasing matrix size. We also observe from Figs.~\ref{fig:inequality} and Figs.~\ref{fig:inequalityscatter}f, that Eq.~(\ref{conj1}) has $f(n)$ growing faster with $n$ than $\sim \sqrt{n}$ even on average. In the case of larger pruned matrices we removed the two expressions $\log((Z_{BP}(\prod_{(i,j)\in{\cal E}}(1-\beta_{ij}))^{-1}\prod_j(1-\sum_i(\beta_{ij})^2))/Z)$ and $\log((2 Z_{BP}(\prod_{i,j}(1-\beta_{ij}))^{-1}\prod_i \beta_{i\Pi(i)}(1-\beta_{i\Pi(i)}))/Z)$. We removed the former because in the case of pruning the resulting $\beta$ is often partially-resolved (with some elements of $\beta$ equal to one) and in this case the inequality does not carry any restriction. We removed the latter because, in the pruned case and for a randomly chosen permutation, it is very likely that at least one element of $\beta$ is zero, making the bound discussed unrestrictive.

\begin{figure}
\centering
\includegraphics[width=0.5\textwidth]{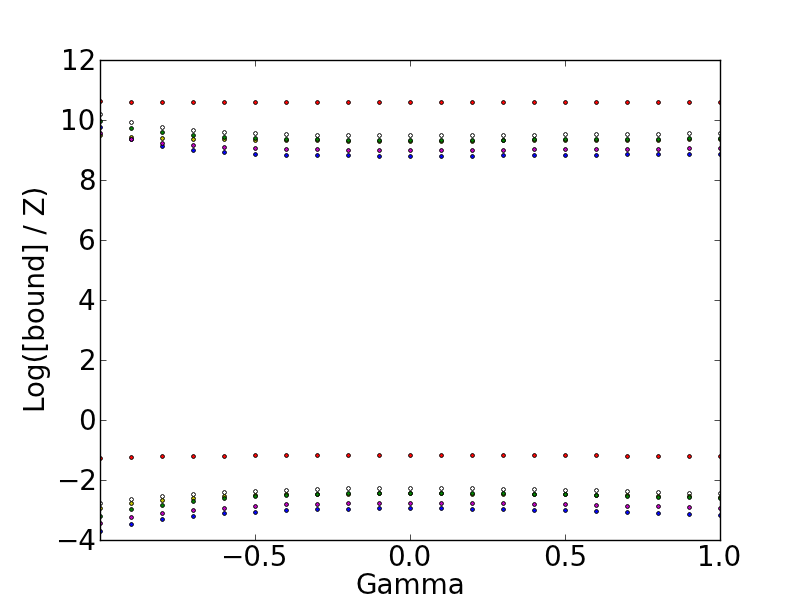}
\caption{
This plot shows $\gamma$-dependence of the gap between upper and lower bounds corresponding to Corollaries \ref{corr:fractional-low}, \ref{corr:fractional-upper}. Here, the bounds were plotted for six different matrices, each generated with   $\lambda_{\mathrm{in}}=\lambda_{\mathrm{out}}=(1,1,1,1)/2$, where the upper and lower bounds are color-coded to indicate that they correspond to the same matrix.}
\label{fig:gamma}
\end{figure}

Fig.~\ref{fig:gamma} shows that the bounds given by the Corollaries \ref{corr:fractional-low}, \ref{corr:fractional-upper} do not depend much on $\gamma$ and that they in practice depend more on matrix size and on peculiarities of individual matrices. There is a slight change for values of $\gamma$ near $-1$, but otherwise the plot is nearly flat, so it seems that unfortunately little tightening of the bounds can be achieved by tweaking $\gamma$. Another noteworthy observation is that a higher upper bound implies a higher lower bound, and vice-versa.

\begin{figure}
\centering
\subfigure[]{\includegraphics[width=0.3\textwidth]{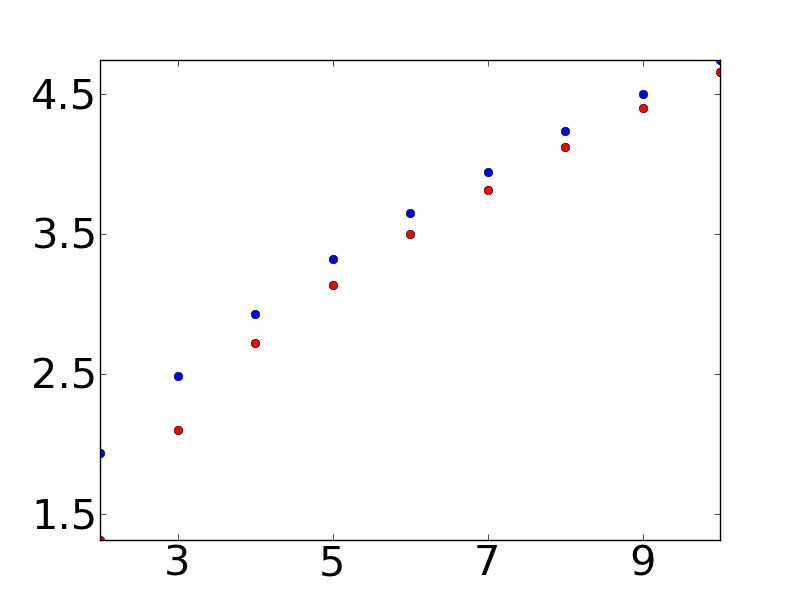}}
\subfigure[]{\includegraphics[width=0.3\textwidth]{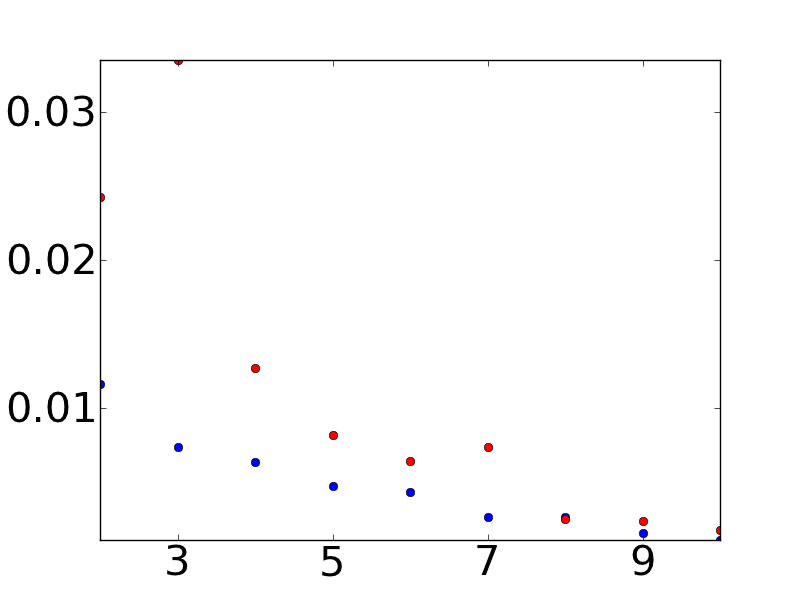}}
\caption{A plot of the average (Subfigure a) and standard deviation (Subfigure b) of $\perm(\beta)/\prod_{i,j}(1-\beta_{ij})^{1- \beta_{ij}}$ (in blue) and $\perm(\beta)/Z_{o-BP}(\beta)$ (in red), where $\beta$ is a doubly stochastic matrix picked from $100$ instances of the random ensemble described in the text, shown as a function of $n$. }
\label{fig:gurvits}
\end{figure}

Fig.~\ref{fig:gurvits} is related to discussions of Corollary \ref{cor:Gurvits} for the permanent of a doubly stochastic matrix. We generate an instance of a doubly stochastic matrix and calculate the respective BP expression in three steps (this is the procedure of \citet{67KS}, also discussed by \citet{09HJ}): (a) generate a non-negative matrix from the $[0;1]$ ensemble; (b) re-scale rows and columns of the matrix iteratively to get a respective doubly stochastic matrix \footnote{The rescaling is a key element of \citet{98LSW}, and we can also think of the procedure as of a version of the $\gamma=0$ iterative algorithm.}; and (c) apply the BP- ($\gamma=-1$) procedure to evaluate the $Z_{o-BP}$ estimate for the resulting doubly stochastic matrix. In agreement with Eq.~(\ref{B1-ext}), the average value of the log corresponding to the BP-lower bound is positive and smaller than the respective expression for the average of the log of the explicit expression on the right-hand side of Eq.~(\ref{B1-ext}). (The hierarchical relation obviously holds as well for any individual instance of the doubly stochastic $\beta$ from the generated ensemble.)
We also observe that the average values of the curves show a tendency to saturate,  while the standard deviation decreases dramatically, suggesting that for large $n$ this random ensemble may be well approximated by either BP or, even more simply, by its explicit lower bound from the right-hand side of Eq.~(\ref{B1-ext}), the latter being in the agreement with the proposal of \citet{11Gur}.

\section{Conclusions and Path Forward}

\label{sec:conclusion}

The main message of this and other related recent papers by \citet{08CKV,09HJ,10CKKVZ,10WC,10Von,11Von,11Gur} is that the BP approach and improvements not only give good heuristics for computing permanents of non-negative matrices,  but also provide theoretical guarantees and thus reliable deterministic approximations. The main highlights of this manuscript are
\begin{itemize}
\item The construction of the fractional approach, parameterized by $\gamma\in [-1;1]$ and interpolating between BP ($\gamma=-1$) and MF ($\gamma=1$) limits.
\item The discovery of the exact relation between the permanent of a non-negative matrix, $\perm (p)$ and the respective fractional expression, $Z_f^{(\gamma)}(p)$, where the latter is computationally tractable.
\item The proof of the continuity and monotonicity of $Z_f^{(\gamma)}(p)$ with $\gamma$,  also suggesting that for some $\gamma_*\in [-1;0]$, $\perm(p)=Z_f^{(\gamma_*)}(p)$.
\item The extension of the list of known BP-based upper and lower bounds for the permanent by their fractional counterparts.
\item The experimental analysis of permanents of different ensembles of interest, including those expressing relations between consecutive images of stochastic flows visualized with particles.
\item Our experimental tests include analysis of the gaps between exact expression for the permanent,  evaluated within the ZDDs technique adapted to permanents,  and the aforementioned BP- and fractional-based lower/upper bounds.
\item The experimental analysis of variations in the special $\gamma$ for different ensembles of matrices suggests the following conclusions. First, the behavior of the special $\gamma$ varies for different ensembles, but the general trend remains the same: as long as there is some element of randomness in the ensemble, the special $\gamma$ decreases as matrix size increases. Second, for each ensemble the behavior of the special $\gamma$ is highly distinctive. For some considered random matrix ensembles the variance decreases quickly with increasing matrix size. All of the above suggest that the fractional approach offers a lot of potential for estimating matrix permanents.
\end{itemize}

We view these results as creating a foundation for further analysis of theoretical and computational problems associated with permanents of large matrices. Of the multitude of possible future problems, we consider the following ones listed below as the most interesting and important:
\begin{itemize}
\item Improving BP and fractional approaches and making the resulting lower and upper bounds tighter.
\item Further analysis of the $\gamma$-dependence, making  theoretical statements for statistics of log-permanents at large $n$ and for different random ensembles.
\item Utilizing the new permanental estimations and bounds for learning flows in the setting of \citet{10CKKVZ}. Combining within the newly introduced fractional approach the $\beta$-optimization with optimization over flow parameters (by analogy with what is done in \citet{10CKKVZ}). Applying the improved technique to various Particle Image Velocimetry (PIV) experiments of interest in fluid mechanics in general,  and specifically to describe spatially smooth multi-pole flows in micro-fluidics,  see e.g. discussion of the most recent relevant experiments in \citet{10DGMPT,10GJG} and references therein \footnote{We are thankful to Eric Lauga for suggesting to us the micro-fluidics experiments as one possible application for the ``learning the flow" BP-based approach.}.
\item Addressing other GM problems of the permanental type, e.g. counting matchings (and not only perfect matchings) on arbitrary graphs (drawing inspiration from \citet{11SMW} generalizing \citet{06BN,08BSS} in the ML setting) and  higher-dimensional matchings, in particular corresponding to matching of paths between multiple consecutive images within the ``learning the flow" setting.
\end{itemize}

\acks{We are thankful to Leonid Gurvits, Yusuke Watanabe, Pascal Vontobel, Vladimir Chernyak, Jonathan Yedidia and Jason Johnson for multiple discussions and very helpful advice, as well as to Leonid Gurvits and Pascal Vontobel for sharing their recent results, \cite{11Gur,11Von} prior to public release. We also very much appreciate the helpful comments and multiple suggested made by the reviewers. ABY acknowledges support of the Undergraduate Research Assistant Program at LANL and he is also grateful to CNLS at LANL for its hospitality. Research at LANL was carried out under the auspices of the National Nuclear Security Administration of the U.S. Department of Energy at Los Alamos National Laboratory under Contract No. DE C52-06NA25396.}

\appendix

\section{Most Probable Perfect Matching}
\label{sec:ML}

This short appendix is introduced to guide the reader through material which is related,  but only indirectly (through physics motivation and historical links), to the main subject of the manuscript.

\subsection{Most Probable Perfect Matching over Bi-Partite Graphs}
\label{subsec:ML}

According to Eq.~(\ref{Z-def}), the permanent can be interpreted as the partition function of a
GM. The partition function represents a weighted counting of the $n!$ perfect matchings. Using ``physics terminology" one says that this perfect matching representation allows to interpret the permanent as the statistical mechanics of perfect matchings (called dimers in the physics literature) over the bi-partite graph. This is statistical mechanics at finite temperatures,  as the partition function represents a (statistical) sum over the perfect matchings.

However,  it is still of interest to discuss (at least in the context of establishing historical links) the  ``zero temperature," or Maximum Likelihood (ML) version of Eq.~(\ref{Z-def})
\begin{eqnarray}
-\log(Z_{ML}(p))=\min_\sigma \sum_{(i,j)\in{\cal E}} \sigma_{ij} \log(1/p_{ij}).
\label{ML}
\end{eqnarray}
According to the logic of \citet{05YFW}, Eq.~(\ref{ML}) can also be stated in the probability space (i.e., in terms of $b(\sigma)$) as
\begin{eqnarray}
-\log(Z_{ML}(p))=\min_{b} \sum_{(i,j)\in{\cal E}} \sum_\sigma \sigma_{ij} \log(1/p_{ij}) b(\sigma).
\label{ML_LP}
\end{eqnarray}
Then, the functional of $b(\sigma)$ which is the object of minimization over beliefs in Eq.~(\ref{ML_LP}) is naturally the ML (zero temperature) version of the FE functional
(\ref{KL}).

By construction, $Z_{ML}(p)\leq Z(p)$ for any $p$. Note also that  Eq.~(\ref{ML}) is a Linear Programming (LP) equation,  but one which at first sight appears intractable, giving an optimization defined over a huge polytope and spanning all the perfect matchings with nonzero probability. For a general GM the LP-ML formulation is indeed intractable, but for the specific problem under consideration (finding the perfect matching over a bipartite graph) the ML-LP problem (\ref{ML_LP}) becomes tractable,  as discussed below in the next subsection. Given classical results from the optimization theory, related to the so-called Hungarian algorithm, by \citet{55Kuh}, and the auction algorithm, by \citet{92Ber}, this special solvability (reduced complexity) of the ML perfect matching problem is not surprising.

\subsection{Linear Programming Relaxation of BP}
\label{subsec:LP}

The Bethe FE (\ref{F_BP}) can be split naturally into the self-energy term and the self-entropy terms (at unit temperature),
$F_{BP}=E_{BP}-S_{BP}$:
\begin{eqnarray}
&& E_{BP}(\beta|p)=-\sum_{(i,j)}\beta_{ij}\log(p_{ij}),\label{E_BP}\\
&& S_{BP}(\beta|p)=\sum_{(i,j)}\left(-\beta_{ij}\log(\beta_{ij})+(1-\beta_{ij})\log(1-\beta_{ij})\right).
\nonumber
\end{eqnarray}
If the entropy term is ignored in Eq.~(\ref{o-BP}) the problem turns into the Linear Programming (LP) formulation of BP
\begin{eqnarray}
-\log(Z_{LP}(p))=\min_{\beta} E_{BP}(\beta).
\label{o-LP}
\end{eqnarray}

One can also arrive at the same LP formulation (\ref{o-LP}) relaxing the original ML-LP setting (\ref{ML_LP}). As shown in \citet{08BSS,08Che}, the relaxation is provably tight for any $p$ ,  i.e., $Z_{LP}(p)=Z_{ML}(p)$, as the resulting matrix of constraints in the LP problem (\ref{o-LP}) describing the doubly stochasticity of $\beta$ is totally uni-modular, so the corners of the respective polytope are in one-to-one correspondence with the perfect matching configurations/corners of the higher-dimensional polytope from Eq.~(\ref{ML_LP}), also in accordance with the Birkhoff-von Neumann theorem by \citet{36Kon,46Bir,53vonNeu}.

\section{Bethe-Free Energy Approach}
\label{sec:BP}

\subsection{Exact BP-based Relations for Permanents}
\label{subsec:BP-exact}

We present here a simple proof of Eq.~(\ref{LC}), essentially following a slightly modified version of what was the main statement of \citet{10WC}.

Consider an interior minimum of the Bethe FE functional (\ref{F_BP}) achieved with a strictly nonzero (for elements with positive $p_{ij}$) doubly stochastic $\beta$. Then, the minimum satisfies Eqs.~(\ref{BP2}),  where $\log(u)$ are respective Lagrangian multipliers.
Weighting the logarithm of Eqs.~(\ref{BP2}) with $\beta$, summing up the result over all the edges, using Eq.~(\ref{F_BP}) and the double stochasticity of $\beta$, one derives
\begin{eqnarray}
&& \sum_{(i,j)\in{\cal E}}\beta_{ij}\log(u_i u_j)=\sum_{i\in{\cal V}_1}\log u_i+\sum_{j\in{\cal V}_2}\log u^j
\nonumber\\ && =
\sum_{(i,j)\in{\cal E}}\left(\beta_{ij}\log(p_{ij}/\beta_{ij})-\beta_{ij}\log(1-\beta_{ij})\right)=
\log Z_{BP}-\sum_{(i,j)\in{\cal E}}\log(1-\beta_{ij}).
\label{u-eqs}
\end{eqnarray}
On the other hand, applying the permanent to both sides of Eq.~(\ref{BP2}) one arrives at
\begin{eqnarray}
\perm(p)=\perm(\beta.*(1-\beta))\left(\prod_{i\in{\cal V}_1}u_i\right)
\left(\prod_{j\in{\cal V}_2}u^j\right).
\label{BP-Perm}
\end{eqnarray}
Combining Eq.~(\ref{u-eqs}) with Eq.~(\ref{BP-Perm}) results in Eq.~(\ref{LC}).

\subsection{Iterative Algorithm(s) for finding solution of BP equations}
\label{sec:BP-iterat}

First of all,  let us recall that according to Proposition \ref{prop:BP-convexity}, (\ref{F_BP}) is convex. However,  as explained above the convexity is not trivial, as it is enforced by global constraints. This lack of convexity of individual edge-local terms in Eq.~(\ref{F_BP}) creates a technical obstacle to finding a valid fixed point of $F_{BP}$, suggesting that an iterative algorithm converging to the fixed point of $F_{BP}$ will be more elaborate than the one discussed below in the MF case.

To find a valid solution of BP in our numerical experiments we use the following practical iterative scheme (heuristics), previously described in \citet{08CKV} (see Eqs.~(7,8) as well as preceding and following explanations):
\begin{align}
 & \forall (i,j):\ \ \beta_{ij}(n\!+\!1)\!=\!\lambda\beta_{ij}(n)\!+\!
 \frac{(1-\lambda)p_{ij}}{p_{ij}+(\sum_k\beta_{kj}(n)/2+\sum_k
\beta_{ik}(n)/2-\beta_{ij}(n))^2 (u_i(n)u^j(n))},
 \label{beta_n}\\
 & \forall i:\ \
 u_i(n+1)=\frac{\sum_k p_{ik} /u^k(n)}{1-\sum_j(\beta_{ij}(n))^2},\quad
\forall j:\ \ u^j(n+1)=\frac{\sum_k p_{kj}/u_k(n)}{1-\sum_i(\beta_{ij}(n))^2},
 \label{U_n}
\end{align}
where the arguments of the $\beta$'s indicate the order of the iterations.
The damping parameter $\lambda$ (typically
chosen $0.4\div0.5$) helps with convergence. To ensure appropriate accuracy for
solutions with $\beta$'s close to zero or unity we also insert a normalization step after Eqs.~(\ref{beta_n}) but prior to Eqs.~(\ref{U_n}), making the following two
transformations subsequently, (a) $\forall (i,j)$: $\beta_{ij}\to
\beta_{ij}/\sum_k\beta_i^k$, and (b) $\forall (i,j)$: $\beta_{ij}\to
\beta_{ij}/\sum_k\beta_k^j$. (The two steps implement an elementary step of the Sinkhorn operation from \citet{09HJ}.)
The algorithm is sensitive to initial values for $\beta$ and $u$. To ensure convergence, one initiates the algorithm with the output of the MF scheme (which converges much better) as described in Appendix \ref{subsec:MF-iterat},  i.e., $\beta(0)=\beta_{MF}$ and $u(0)=v_{MF}$. Numerical experiments show that this procedure always converges to an interior stationary point of the BFE (\ref{F_BP}), when one exists and is not degenerate.  In the special cases when the solution is on the boundary it seemed to converge there as well, but we did not study this systematically to make a definitive statement.

Note that the algorithm presented above is certainly not the only option one can use to find a doubly stochastic solution of BP Eqs.~(\ref{BP2}).  In fact,  the standard Sum-Product Algorithm (SPA) of \citet{05YFW},  stated for the problem of computing the permanent in \citet{10CKKVZ}, is a serious competitor,  which according to Theorem 32 of \citet{11Von} always converges to the minimum of the Bethe FE. Future work is required to compare the convergence speed of the two algorithms.

\section{Mean-Field (Fermi) Approach}

\label{sec:MF}

\subsection{Exact MF-based Relations for Permanents}
\label{subsec:MF-exact}

We present here a simple proof of Eq.~(\ref{MF-exact}), essentially following the logic of what was described above for BP in Appendix \ref{subsec:BP-exact}.

Weighting the logarithm of Eqs.~(\ref{MF3}) with that doubly stochastic $\beta$ which minimizes Eq.~(\ref{MF-exact}),  summing the result over all the edges,  and making use of Eqs.~(\ref{MF3},\ref{MF1}), one derives
\begin{eqnarray}
&& \sum_{(i,j)\in{\cal E}}\beta_{ij}\log(v_i v_j)=\sum_{i\in{\cal V}_1}\log v_i+\sum_{j\in{\cal V}_2}\log v^j
\nonumber\\ && =
\sum_{(i,j)\in{\cal E}}\left(\beta_{ij}\log(p_{ij}/\beta_{ij})+\beta_{ij}\log(1-\beta_{ij})\right)=
\log Z_{MF}(p)+\sum_{(i,j)\in{\cal E}}\log(1-\beta_{ij}).
\label{v-eqs}
\end{eqnarray}
On the other hand, applying the permanent to both sides of Eq.~(\ref{MF3}) one arrives at
\begin{eqnarray}
\perm(p)=\perm(\beta./(1-\beta))\left(\prod_{i\in{\cal V}_1}v_i\right)
\left(\prod_{j\in{\cal V}_2}v^j\right).
\label{MF-Perm}
\end{eqnarray}
Combining Eq.~(\ref{v-eqs}) with Eq.~(\ref{MF-Perm}) results in Eq.~(\ref{MF-exact}).

\subsection{Iterative Scheme for Solving Mean-Field Equations}
\label{subsec:MF-iterat}

An efficient heuristic way to find a (unique) solution of the MF system of Eqs.~(\ref{MF3}) for doubly stochastic $\beta$ is to initialize with $v_i(0)=v^j(0)=1$ and iterate according to
\begin{align}
& \beta_{ij}(n+1)=\frac{p_{ij}}{p_{ij}+v_i(n)v^j(n)},
\label{MF_it_1}
\\
& v_i(n+1)=v_i(n) \sum_j \beta_{ij}(n),\quad v^j(n+1)=v^j(n) \sum_i \beta_{ij}(n),
\label{MF_it_2}
\end{align}
until the tolerance $\delta>\max(\mbox{abs}(\beta(n+1)-\beta(n)))$ is met.

\section{Fractional Approach}

\label{sec:fractional}

\subsection{Exact Relations for Permanents}
\label{subsec:fractional-exact}

We present here a simple proof of Eq.~(\ref{m-exact}),  which is a direct generalization of what was discussed above in Appendices  \ref{subsec:BP-exact},\ref{subsec:MF-exact}.

Weighting the logarithm of Eqs.~(\ref{m3}) with that doubly stochastic $\beta$ which minimizes Eq.~(\ref{m-exact}),  summing the result over all the edges,  and making use of Eqs.~(\ref{m3},\ref{m1}), one derives
\begin{align}
& \sum_{(i,j)\in{\cal E}}\beta_{ij}\log(w_i w_j)=\sum_{i\in{\cal V}_1}\log w_i+\sum_{j\in{\cal V}_2}\log w^j
\nonumber\\ & =
\sum_{(i,j)\in{\cal E}}\left(\beta_{ij}\log(p_{ij}/\beta_{ij})+\gamma\beta_{ij}\log(1-\beta_{ij})\right)\nonumber\\
&=
\log Z_f^{(\gamma)}(\beta|p)+\gamma\sum_{(i,j)\in{\cal E}}\log(1-\beta_{ij}).
\label{w-eqs}
\end{align}
On the other hand, applying the permanent to both sides of Eq.~(\ref{m3}) one arrives at
\begin{eqnarray}
\perm(p)=\perm(\beta./(1-\beta).^\gamma)\left(\prod_{i\in{\cal V}_1}w_i\right)
\left(\prod_{j\in{\cal V}_2}w^j\right).
\label{m-Perm}
\end{eqnarray}
Combining Eq.~(\ref{w-eqs}) with Eq.~(\ref{m-Perm}) results in Eq.~(\ref{m-exact}).

\subsection{Iterative Scheme for Solving Fractional Equations}
\label{subsec:fractional-iterat}

All edge-local terms in the fractional functional (\ref{m1}) are convex in $\beta\in[0;1]$ for $\gamma>0$,  while for negative $\gamma$ the edge-term convexity holds only when all elements of $\beta$ are smaller than a threshold $\beta_c\geq 1/2$, which is a solution of $\beta_c\log(\beta_c)=-\gamma(1-\beta_c)\log(1-\beta_c)$. This suggests different iterative schemes for positive and negative $\gamma$.

When $\gamma>0$ we use the following modification of the MF scheme (\ref{MF_it_1},\ref{MF_it_2}):
\begin{align*}
& \beta_{ij}(n+1)=\frac{p_{ij}(1-\beta_{ij}(n))^{\gamma-1}}{p_{ij}(1-\beta_{ij}(n))^{\gamma-1}+w_i(n)w^j(n)}, \\
& w_i(n+1)=w_i(n) \sum_j \beta_{ij}(n),\quad w^j(n+1)=w^j(n) \sum_i \beta_{ij}(n),
\end{align*}

In the case of $\gamma\leq 0$ we use the following modification of the BP scheme
\begin{align}
 & \forall (i,j):\ \ \beta_{ij}(n+1)=\lambda\beta_{ij}(n)\label{rho_n}\\
 & +
 \frac{(1-\lambda)p_{ij}(1+\beta_{ij}(n))^{1+\gamma}}{p_{ij}(1+\beta_{ij}(n))^{1+\gamma}+(\sum_k\beta_k^j(n)/2+\sum_k
\beta_i^k(n)/2-\beta_{ij}(n))^2 (w_i(n)w^j(n))},\nonumber
 \\
 & \forall i:\ \
 w_i(n+1)=\frac{\sum_k p_{ik} (1+\beta_{ik}(n))^{1+\gamma}/w^k(n)}{1-\sum_j(\beta_{ij}(n))^2},\nonumber\\
& \forall j:\ \ w^j(n+1)=\frac{\sum_k p_{kj}(1+\beta_{kj}(n))^{1+\gamma}/w_k(n)}{1-\sum_i(\beta_{ij}(n))^2},
 \label{w_n}
\end{align}
where the arguments of the $\beta$'s indicate the order of the iterations.
The damping parameter $\lambda$ (typically
chosen $0.4\div0.5$) helps with convergence. To ensure appropriate accuracy for
solutions with $\beta$'s close to zero or unity we also insert a normalization step
after Eqs.~(\ref{rho_n}) but prior to Eqs.~(\ref{w_n}), making the following two
transformations consequently,
\begin{eqnarray*}
(a)\quad \forall (i,j):\quad \beta_{ij}\to
\beta_{ij}/\sum_k\beta_i^k,\\
 (b)\quad \forall (i,j):\quad \beta_{ij}\to
\beta_{ij}/\sum_k\beta_k^j.
\end{eqnarray*}
The algorithm is sensitive to initial values for $\beta$ and $w$. To ensure convergence, we initiate the algorithm with the output of the MF scheme (which converges much more easily) described in Appendix \ref{subsec:MF-iterat},  i.e., $\beta(0)=\beta_{MF}$ and $w(0)=v_{MF}$. Numerical experiments show that this procedure converges to a stationary point of the fractional FE (\ref{m1}). We also verified that the iterative scheme designed for $\gamma<0$ converges in the $\gamma>0$ case, even thought  the former scheme is obviously faster.

Note that fractional version of the standard Sum-Product Algorithm (SPA) can be developed.  It is also natural to expect,  in view of the general convexity of the fractional FE discussed in the main body of the text,  that  there exists a provably convergent version of the SPA.  It will be important to design such a convergent $\gamma$-SPA in the future and to compare its practical performance against one of the heuristics described above.

\section{BP Gives Lower Bound on the Permanent}
\label{sec:BP_low}

Here we give our version of the proof of the lower bound (\ref{B1}). First of all, in the case when the Bethe FE reaches its minimum in the interior of the domain, i.e., at $\beta\in{\cal B}_p$,  Eq.~(\ref{B1}) follows directly from the main result of \citet{10WC}, i.e.,  Eq.~(\ref{LC}), and Schijver's inequality (\ref{Sch-ineq}). Therefore, according to explanations of Section \ref{subsubsec:o-BP}, we only need to analyze the case when the minimum of the Bethe FE is a partially resolved solution, with a $\beta$ which can be split by appropriate permutations of rows and columns of the matrix into a perfect matching block (corresponding to a corner of the respective projected polytope), the block with all elements smaller than unity and nonzero unless the respective element of $p$ is zero (thus lying in the interior of the respective subspace), and all cross elements of $\beta$ (between the blocks) equal to zero. Then, $Z_{o-BP}$ for such a partially resolved solution is split into the product of two contributions,  $Z_{o-BP}=Z_{pm}\cdot Z_{int}$,  where
$Z_{pm}$ corresponds to the perfect matching block,  and $Z_{int}$ corresponds to the interior block.
In fact,  $Z_{pm}$ is equal to the weighted perfect matching block of $p$ and $-\log(Z_{int})$
corresponds to the minimum of the Bethe FE computed for the interior block of $p$. On the other hand the full partition function, $Z$,  can be bounded from below by the product $Z\geq Z_1\cdot Z_2$,  where $Z_1$ and $Z_2$ are permanents of the first and second blocks of the original matrix $p$. (Thus contributions of all the cross-terms of $p$ into $Z$ are ignored.) However,  $Z_1\geq Z_{pm}$, as counting only one perfect matching (and ignoring others), and $Z_2\geq Z_{int}$ in accordance to what was already shown above for any minimum of Bethe FE achieved in the interior of the respective domain.

\section{Pruning of the Matrix}
\label{sec:pruning}

Computing the permanent of sufficiently dense matrixes exactly with the ZDD approach explained in Appendix \ref{sec:ZDD} is infeasible for $n>30$. To overcome this difficulty we choose to sparsify dense matrices generated in one of our experimental ensembles, removing their less significant entries in the following steps.
First, we use LP, described in Appendix \ref{subsec:LP}, to find the permutation correspondent to the maximum perfect matching. To avoid getting a zero permanent in the result, we include all components of the maximum perfect matching permutation in the pruned matrix. Second, we consider every other entry of the matrix (not contributing the maximum perfect matching) and keep it in the matrix only if it is included in a perfect matching which is close to the maximum perfect matching, i.e., the two permutations share all but a few of their entries and ratio of their weighted contributions (in the permanent) is larger than a pre-defined value.
Then,  we act according to either of the two strategies, both of which are explored in this manuscript. One strategy is to include all permutations whose products are more than a given fraction of the main permutation. This method will tend to reduce the fluctuations in the error of the pruned matrix (i.e., will reduce the variation in $Z_{\mathrm{pruned}} / Z_{\mathrm{original}}$). The other method is to always prune a set fraction of entries from the matrix, and prune them in order of decreasing value as determined by the above criterion. This method will reduce the fluctuations in the runtime of the algorithm.

\section{Zero-suppressed Binary Decision Diagrams (ZDD) method}
\label{sec:ZDD}

Zero-suppressed Binary Decision Diagrams, or ZDD, are a tool useful for representing combinatorial problems. The concept was introduced by
Shin-Ichi Minato in 1993 \cite{93Minato}.
The idea of ZDD is as follows: if one defines a combinatorial problem to be a function of many variables, each taking values in \{0, 1\}, with the value of the function itself being also in \{0, 1\}, then those sequences of inputs that lead to unity can be thought of as ``solutions" to the problem. Furthermore, each solution can be described in terms of the input variables within it that are equal to unity. The problem, then, can be described as being a ``family of sets," or set of sets, where the family is of all solutions to the problem and each set within the family is the set of input variables whose value is 1 in that solution.

To give an example of the ``family of sets" concept, consider the $\overline{\mbox{XOR}}$ function, which returns 1 if and only if the inputs are equal. This function can also be represented as the family of sets $\{\emptyset, \{1, 2\}\}$, where 1 and 2 correspond to inputs 1 and 2 to the function, because if the function is to have value 1 then either both inputs must be equal to 1 or neither must be.
Once this has been understood, it is best to see the ZDD as nothing more than a concise representation of this family of sets, since the family can get quite cumbersome for problems with many solutions and many variables.
Note that this system of representing problems provides the greatest improvement when there are few solutions, and when the solutions themselves are sparse, since the family of sets is then small. Correspondingly, ZDD are most efficient under these conditions.

The actual format of a ZDD is that of a directed tree of nodes, with each node having a directed edge to two other nodes. Each edge emanating from a node has an identity, in that it is either a ``HI" branch or a ``LO" branch, and of the two edges emanating from each node, there must be exactly one ``HI" branch and one ``LO" branch. Each node also has an identity, a number from 1 to $n$ if there are $n$ inputs to the combinatorial problem. The tree must contain one or two special nodes, or ``sinks", one of which is the ``True" sink, and optionally the ``False" sink. We also introduce the conventions that nodes can only point to nodes of higher identity than themselves and that no two nodes can be identical in both their identity and their LO and HI pointers.

Each node in a ZDD represents a choice about the variable the node identifies. If one begins at the top node of a ZDD, taking the HI branch represents including the variable represented by the node's identity in a prospective solution, and taking the LO branch represents not including that variable. If a LO or HI branch points to the True sink, that implies that a solution is reached if and only if all variables with identity greater than the current node identity are not included. If a LO branch points to the False sink, that implies that no solution is possible given the choices made previously. Interestingly, the constraints introduced in the paragraph previous to this one imply that a HI branch can never point to the False sink.

\begin{figure}
\centering
\subfigure[]{\includegraphics[width=0.3\textwidth]{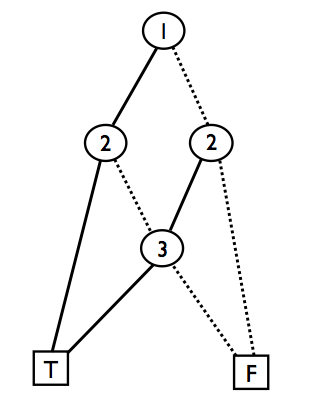}\label{fig:zdd1}}
\subfigure[]{\includegraphics[width=0.3\textwidth]{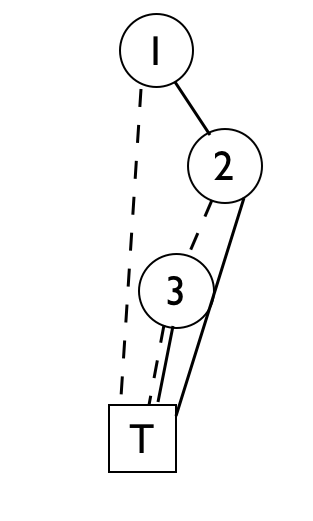}\label{fig:zdd2}}
\caption{Two simple ZDD diagrams discussed in the text.}
\end{figure}

ZDD are best understood with examples. The first example, also illustrated in Fig.~\ref{fig:zdd1}, is of the ZDD for the exactly-two function of three variables, in other words, the function that returns 1 if exactly two of its three inputs have value 1 and 0 otherwise. It can also be described as the family of sets \{\{1, 2\}, \{1, 3\}, \{2, 3\}\}. Here, a dotted line denotes a LO branch and a normal line denotes a HI branch. Furthermore, the T and F symbols denote the True and False sinks, respectively, and the numbers inside each node refer to the nodes' identities (the variables that they represent). Our second example,  shown in Fig.~\ref{fig:zdd2}, represents the family of sets $\{\emptyset, \{1\}, \{1, 2\}, \{1, 3\}\}$. Note the absence of a False sink. Note, also, the fact that a node's HI and LO branches need not point to the different locations. A more in-depth exploration of the ZDD concept can be found in \citet{Knuth:2009:ACP:v4f1}.

Once the basic concept of ZDD is introduced, one can use it for solving various combinatorial problems,  e.g. to represent a permanent as a ZDD in order to use the method. When we apply ZDD to the computations of permanent, we classify each entry of the matrix as either zero or nonzero. Then, we define a variable for each nonzero entry in the matrix. Each solution of our resulting ZDD will represent a possible permutation, meaning a set of entries in the matrix such that exactly one entry in each row and column is included in the set. There is a recursive algorithm, suggested in \citet{Knuth:2009:ACP:v4f1}, that allows for efficient counting of the solutions of the ZDD. The algorithm is simple: the number of solutions of a ZDD rooted at a node is equal to the sum of the numbers of solutions of the ZDD rooted at the HI and LO children of that node. The True sink is defined as having 1 solution, and the False sink as having 0. Note that the number of solutions of a ZDD representing a matrix is equal to the permanent of the corresponding to $0-1$ matrix, with each 1 corresponding a nonzero entry.

In order to find the permanent of matrices that are not 0-1 matrices, only a small modification is necessary. Instead of purely counting solutions of the ZDD, we do a weighted count, where the weighted number of solutions of a ZDD rooted at a node is equal to the value of the corresponding matrix entry times the weighted number of solutions at the HI child added to the weighted number of solutions at the LO child. In other words, if we are considering a node $n$ with children HI and LO whose identity corresponds to a matrix entry of nonzero value $v$, then
$$\mbox{WeightedCount}(n) = v \cdot \mbox{WeightedCount}(\mbox{HI}) + \mbox{WeightedCount}(\mbox{LO}).$$
The WeightedCount of the root node of the ZDD will be equal to the permanent of the corresponding matrix.

This leaves the question of how to build the ZDD from the matrix. This is done using Knuth's ``melding" algorithm. The algorithm is somewhat complex and will not be described here, but it is described in detail in \citet{Knuth:2009:ACP:v4f1}. The melding algorithm is an efficient and systematic method for constructing larger ZDD out of the logical combination of smaller ones. The smallest ZDD being melded together using Knuth's algorithm are ZDD representing the ``exactly-one" constraint for each row and column of the matrix; in other words, they are constraints requiring exactly one matrix entry in every row and column to be included in a permutation which will be a ``solution" to our problem.

\section{Comparison of Ryser's formula with the ZDD-based method}
\label{sec:Ryser_vs_ZDD}

\begin{figure}%
\centering
\subfigure[pruning 40\%]{\includegraphics[width=0.3\textwidth]{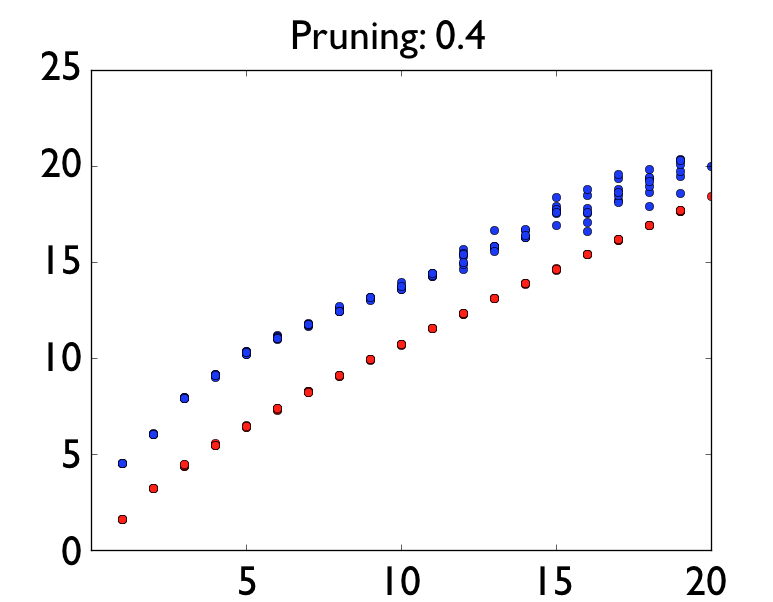}}
\subfigure[pruning 60\%]{\includegraphics[width=0.3\textwidth]{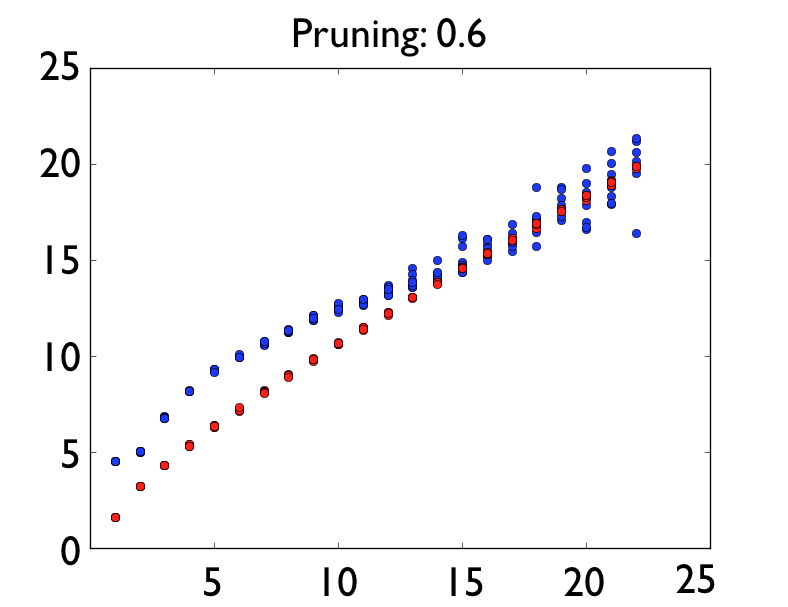}}
\subfigure[pruning 80\%]{\includegraphics[width=0.3\textwidth]{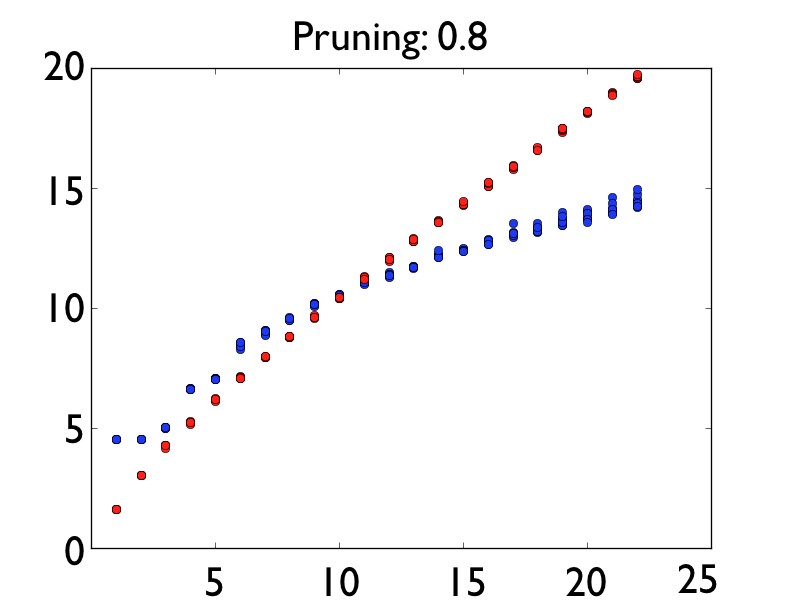}}
\caption{Scatter plot of the number of memory accesses required to exactly compute the permanent of a sparse matrix for instances with different degree of pruning. Red and blue dots mark results of the Ryser formula and of the ZDD-based method, respectively.}
\label{fig:Ryzer_vs_ZDD}
\end{figure}

As part of our experiments we compared the speed of Ryser's formula with the speed of the ZDD-based method by counting memory accesses in each of the two algorithms in order to fairly compare them. We found that the values we got for memory accesses were strongly correlated with the actual speed of the algorithm. We found that for very dense matrices, Ryser's formula is faster, but for sparser matrices the ZDD-based method is faster. We performed experiments with matrices that were 20\%, 40\%, and 60\% sparse in order to get a good idea of the point where the ZDD-based method starts outperforming Ryser's formula. (Naturally, with no pruning, Ryser's formula outperforms the ZDD-based method significantly.)

As can be seen from Fig.~\ref{fig:Ryzer_vs_ZDD}, the ZDD-based method begins outperforming Ryser's formula when matrices are around 60\% sparse.

\bibliographystyle{unsrt}
\bibliography{new_permanent}

\begin{thebibliography}{44}
\providecommand{\natexlab}[1]{#1}
\providecommand{\url}[1]{\texttt{#1}}
\expandafter\ifx\csname urlstyle\endcsname\relax
  \providecommand{\doi}[1]{doi: #1}\else
  \providecommand{\doi}{doi: \begingroup \urlstyle{rm}\Url}\fi

\bibitem[Bayati and Nair(2006)]{06BN}
M.~Bayati and C.~Nair.
\newblock A rigorous proof of the cavity method for counting matchings.
\newblock In \emph{Proc. 44th Allerton Conf. on Communications, Control, and
  Computing}, 2006.

\bibitem[Bayati et~al.(2008)Bayati, Shah, and Sharma]{08BSS}
M.~Bayati, D.~Shah, and M.~Sharma.
\newblock {Max-product for maximum weight matching: {C}onvergence, correctness,
  and {LP} duality}.
\newblock \emph{IEEE Transactions on Information Theory}, 54\penalty0
  (3):\penalty0 1241--1251, 2008.

\bibitem[Bertsekas(1992)]{92Ber}
D.P. Bertsekas.
\newblock {Auction algorithms for network flow problems: A tutorial
  introduction}.
\newblock \emph{Comput. Optimiz. Applic.}, 1:\penalty0 7–--66, 1992.

\bibitem[Bethe(1935)]{35Bet}
H.A. Bethe.
\newblock {Statistical theory of superlattices}.
\newblock \emph{Proceedings of Royal Society of London A}, 150:\penalty0 552,
  1935.

\bibitem[Birkhoff(1946)]{46Bir}
G.~Birkhoff.
\newblock {Three observations on linear algebra}.
\newblock \emph{Univ. Nac. Tacum´an Rev. Ser. A}, 5:\penalty0 147–--151, 1946.

\bibitem[Chertkov(2008)]{08Che}
M.~Chertkov.
\newblock {Exactness of belief propagation for some graphical models with
  loops}.
\newblock \emph{Journal of Statistical Mechanics: Theory and Experiment},
  2008\penalty0 (10):\penalty0 P10016 (12pp), 2008.
\newblock URL \url{http://stacks.iop.org/1742-5468/2008/P10016}.

\bibitem[Chertkov and Chernyak(2006{\natexlab{a}})]{06CCa}
M.~Chertkov and V.~Y. Chernyak.
\newblock {Loop calculus in statistical physics and information science}.
\newblock \emph{Phys. Rev. E}, 73\penalty0 (6):\penalty0 065102, Jun
  2006{\natexlab{a}}.
\newblock \doi{10.1103/PhysRevE.73.065102}.

\bibitem[Chertkov and Chernyak(2006{\natexlab{b}})]{06CCb}
M.~Chertkov and V.Y. Chernyak.
\newblock {Loop series for discrete statistical models on graphs}.
\newblock \emph{Journal of Statistical Mechanics: Theory and Experiment},
  2006\penalty0 (06):\penalty0 P06009, 2006{\natexlab{b}}.
\newblock URL \url{http://stacks.iop.org/1742-5468/2006/i=06/a=P06009}.

\bibitem[Chertkov et~al.(2008)Chertkov, Kroc, and Vergassola]{08CKV}
M.~Chertkov, L.~Kroc, and M.~Vergassola.
\newblock {Belief Propagation and Beyond for Particle Tracking}.
\newblock \emph{arxiv}, abs/0806.1199, 2008.

\bibitem[{Chertkov} et~al.(2010){Chertkov}, {Kroc}, {Krzakala}, {Vergassola},
  and {Zdeborova}]{10CKKVZ}
M.~{Chertkov}, L.~{Kroc}, F.~{Krzakala}, M.~{Vergassola}, and L.~{Zdeborova}.
\newblock {Inference in particle tracking experiments by passing messages
  between images}.
\newblock \emph{Proceedings of the National Academy of Science}, 107:\penalty0
  7663--7668, April 2010.
\newblock \doi{10.1073/pnas.0910994107}.

\bibitem[Cseke and Heskes(2011)]{11CH}
B.~Cseke and T.~Heskes.
\newblock {Properties of Bethe Free Energies and Message Passing in Gaussian
  Models}.
\newblock \emph{J. Artif. Intell. Res. (JAIR)}, 41:\penalty0 1--24, 2011.

\bibitem[Drescher et~al.(2010)Drescher, Goldstein, Michel, Polin, and
  Tuval]{10DGMPT}
K.~Drescher, R.~E. Goldstein, N.~Michel, M.~Polin, and I.~Tuval.
\newblock {Direct Measurement of the Flow Field around Swimming
  Microorganisms}.
\newblock \emph{Phys. Rev. Lett.}, 105\penalty0 (16):\penalty0 168101, Oct
  2010.
\newblock \doi{10.1103/PhysRevLett.105.168101}.

\bibitem[Egorychev(1981)]{81Ego}
P.V. Egorychev.
\newblock {Proof of the van der {W}aerden conjecture for permanents}.
\newblock \emph{Siberian Mathematical Journal}, 22\penalty0 (6):\penalty0
  854--859, 1981.
\newblock URL \url{http://www.springerlink.com/content/k692377516k1x778/}.

\bibitem[Engel and Schneider(1973)]{73ES}
G.~M. Engel and H.~Schneider.
\newblock {Inequalities for Determinants and Permanents}.
\newblock \emph{Linear and Multilinear Algebra}, 1:\penalty0 187--201, 1973.

\bibitem[Falikman(1981)]{81Fal}
D.I. Falikman.
\newblock {Proof of the van der Waerden conjecture regarding the permanent of a
  doubly stochastic matrix}.
\newblock \emph{Mathematical Notes}, 29\penalty0 (6):\penalty0 475--479, 1981.
\newblock URL \url{http://www.springerlink.com/content/h41162g677317110/}.

\bibitem[Gallager(1963)]{63Gal}
R.G. Gallager.
\newblock \emph{{{L}ow-{D}ensity {P}arity-{C}heck codes}}.
\newblock MIT Press, Cambridge, MA, 1963.

\bibitem[Guasto et~al.(2010)Guasto, Johnson, and Gollub]{10GJG}
J.~S. Guasto, K.~A. Johnson, and J.~P. Gollub.
\newblock {Oscillatory Flows Induced by Microorganisms Swimming in Two
  Dimensions}.
\newblock \emph{Phys. Rev. Lett.}, 105\penalty0 (16):\penalty0 168102, Oct
  2010.
\newblock \doi{10.1103/PhysRevLett.105.168102}.

\bibitem[Gurvits(2008)]{08Gur}
L.~Gurvits.
\newblock {Van der Waerden/Schrijver-Valiant like conjectures and stable (aka
  hyperbolic) homogeneous polynomials: one theorem for all}.
\newblock \emph{Electronic Journal of Combinatorics}, 15:\penalty0 R66, 2008.
\newblock URL
  \url{http://www.emis.ams.org/journals/EJC/Volume_15/PDF/v15i1r66.pdf}.

\bibitem[{Gurvits}(2011)]{11Gur}
L.~{Gurvits}.
\newblock {Unharnessing the power of Schrijver's permanental inequality}.
\newblock \emph{ArXiv e-prints}, June 2011.

\bibitem[Huang and Jebara(2009)]{09HJ}
B.~Huang and T.~Jebara.
\newblock {Approximating the Permanent with Belief Propagation,
  arxiv:0908.1769}, 2009.
\newblock URL \url{http://arxiv.org/abs/0908.1769}.

\bibitem[Huber and Law(2008)]{08HL}
M.~Huber and J.~Law.
\newblock Fast approximation of the permanent for very dense problems.
\newblock \emph{in: SODA 08: Proc. 19th ACM-SIAM Sympos. on Discrete
  Algorithms}, page 681–689, 2008.

\bibitem[Jerrum et~al.(2004)Jerrum, Sinclair, and Vigoda]{04JSV}
M.~Jerrum, A.~Sinclair, and E.~Vigoda.
\newblock {A polynomial-time approximation algorithm for the permanent of a
  matrix with nonnegative entries}.
\newblock \emph{J. ACM}, 51\penalty0 (4):\penalty0 671--697, 2004.
\newblock ISSN 0004-5411.
\newblock \doi{http://doi.acm.org/10.1145/1008731.1008738}.

\bibitem[K\H{o}nig(1936)]{36Kon}
D.~K\H{o}nig.
\newblock \emph{{Theorie der endlichen und unendlichen Graphen}}.
\newblock Akademische Verlags Gesellschaft, Leipzig, 1936.

\bibitem[Knopp and Sinkhorn(1967)]{67KS}
P.~Knopp and R.~Sinkhorn.
\newblock {Concerning nonnegative matrices and doubly stochastic matrices}.
\newblock \emph{Pacific J. Math.}, 21\penalty0 (2):\penalty0 343--348, 1967.

\bibitem[Knuth(2009)]{Knuth:2009:ACP:v4f1}
D.~E. Knuth.
\newblock \emph{{The Art of Computer Programming, Volume 4, Fascicle 1: Bitwise
  Tricks \& Techniques; Binary Decision Diagrams}}.
\newblock Addison-Wesley Professional, 12th edition, 2009.
\newblock ISBN 0321580508, 9780321580504.

\bibitem[Kuhn(1955)]{55Kuh}
H.~W. Kuhn.
\newblock {The Hungarian Method for the assignment problem}.
\newblock \emph{Naval Research Logistics Quarterly}, 2:\penalty0 83–97, 1955.

\bibitem[Laurent and Schrijver(2010)]{09LS}
M.~Laurent and A.~Schrijver.
\newblock {On {L}eonid {G}urvits's proof for permanents}.
\newblock \emph{American Mathematical Monthly}, 117\penalty0 (10), 2010.
\newblock URL \url{http://homepages.cwi.nl/~lex/files/perma5.pdf}.

\bibitem[Linial et~al.(1998)Linial, Samorodnitsky, and Wigderson]{98LSW}
N.~Linial, A.~Samorodnitsky, and A.~Wigderson.
\newblock A deterministic strongly polynomial algorithm for matrix scaling and
  approximate permanents.
\newblock In \emph{Proceedings of the thirtieth annual ACM symposium on Theory
  of computing}, STOC '98, pages 644--652, New York, NY, USA, 1998. ACM.
\newblock ISBN 0-89791-962-9.
\newblock \doi{10.1145/276698.276880}.
\newblock URL \url{http://doi.acm.org/10.1145/276698.276880}.

\bibitem[Minato(1993)]{93Minato}
S.~Minato.
\newblock {Zero-Suppressed {BDD}s for Set Manipulation in Combinatorial
  Problems}.
\newblock In \emph{{Design Automation, 1993. 30th Conference on}}, pages 272 --
  277, 1993.
\newblock \doi{10.1109/DAC.1993.203958}.

\bibitem[Pearl(1988)]{88Pea}
J.~Pearl.
\newblock \emph{{Probabilistic Reasoning in Intelligent Systems: Networks of
  Plausible Inference}}.
\newblock San Francisco: Morgan Kaufmann Publishers, Inc., 1988.

\bibitem[Peierls(1936)]{36Pei}
H.A. Peierls.
\newblock {Ising's model of ferromagnetism}.
\newblock \emph{Proceedings of Cambridge Philosophical Society}, 32:\penalty0
  477--481, 1936.

\bibitem[Ryser(1963)]{63Rys}
H.~J. Ryser.
\newblock Combinatorial mathematics.
\newblock \emph{The Carus Mathematical Monographs}, 14, 1963.

\bibitem[Sanghavi et~al.(2011)Sanghavi, Malioutov, and Willsky]{11SMW}
S.~Sanghavi, D.~Malioutov, and A.~Willsky.
\newblock Belief propagation and lp relaxation for weighted matching in general
  graphs.
\newblock \emph{Information Theory, IEEE Transactions on}, 57\penalty0
  (4):\penalty0 2203 --2212, april 2011.
\newblock ISSN 0018-9448.
\newblock \doi{10.1109/TIT.2011.2110170}.

\bibitem[Schrijver(1998)]{98Sch}
A.~Schrijver.
\newblock {Counting 1-Factors in Regular Bipartite Graphs}.
\newblock \emph{Journal of Combinatorial Theory}, 72:\penalty0 122--135, 1998.
\newblock \doi{10.1006/jctb.1997.1798}.

\bibitem[TheCodeProject()]{Ryser_code}
TheCodeProject.
\newblock
  \emph{\url{http://www.codeproject.com/KB/applications/RyserPermanent.aspx}}.
\newblock {Computes Permanent of a Matrix with Ryser's Algorithm}.

\bibitem[Valiant(1979)]{79Val}
L.G. Valiant.
\newblock The complexity of computing the permanent.
\newblock \emph{Theoretical Computer Science}, 8:\penalty0 189–201, 1979.

\bibitem[{van der Waerden}(1926)]{26vdW}
{van der Waerden}.
\newblock \emph{[Aufgabe] 45, Jahresbericht der Deutschen
  Mathematiker-Vereinigung}, 35:\penalty0 117, 1926.

\bibitem[von Neumann(1953)]{53vonNeu}
J.~von Neumann.
\newblock A certain zero-sum two-person game equivalent to an optimal
  assignment problem.
\newblock \emph{Ann. Math. Studies}, 28:\penalty0 5--12, 1953.

\bibitem[Vontobel(2011)]{11Von}
P.~O. Vontobel.
\newblock {The Bethe Permanent of a Non-Negative Matrix}.
\newblock \emph{arxiv:1107.4196}, 2011.

\bibitem[Vontobel(2010)]{10Von}
P.O. Vontobel.
\newblock {The {B}ethe permanent of a non-negative matrix}.
\newblock In \emph{{Communication, Control, and Computing (Allerton), 2010 48th
  Annual Allerton Conference on}}, pages 341 --346, 29 2010-oct. 1 2010.
\newblock \doi{10.1109/ALLERTON.2010.5706926}.

\bibitem[Watanabe and Chertkov(2010)]{10WC}
Y.~Watanabe and M.~Chertkov.
\newblock {Belief propagation and loop calculus for the permanent of a
  non-negative matrix}.
\newblock \emph{Journal of Physics A: Mathematical and Theoretical},
  43\penalty0 (24):\penalty0 242002, 2010.
\newblock URL \url{http://stacks.iop.org/1751-8121/43/i=24/a=242002}.

\bibitem[Wiegerinck and Heskes(2003)]{Wiegerinck_Heskes_2003}
W.~Wiegerinck and T.~Heskes.
\newblock Fractional belief propagation.
\newblock \emph{Advances in Neural Information Processing Systems 15},
  12:\penalty0 438--445, 2003.

\bibitem[Yedidia(2009)]{09YedA}
A.B. Yedidia.
\newblock {Counting Independent Sets and Kernels of Regular Graphs}.
\newblock \emph{arxiv}, abs/0910.4664, 2009.

\bibitem[Yedidia et~al.(2005)Yedidia, Freeman, and Weiss]{05YFW}
J.S. Yedidia, W.T. Freeman, and Y.~Weiss.
\newblock Constructing free-energy approximations and generalized belief
  propagation algorithms.
\newblock \emph{Information Theory, IEEE Transactions on}, 51\penalty0
  (7):\penalty0 2282 -- 2312, {J}uly 2005.
\newblock ISSN 0018-9448.
\newblock \doi{10.1109/TIT.2005.850085}.

\end{thebibliography}

\end{document}